\documentclass[journal,comsoc]{IEEEtran}

\usepackage[T1]{fontenc}
\usepackage{graphicx}
\usepackage{subfigure} 
\usepackage{cite}

\usepackage{amsmath}
\newtheorem{assumption}{\bf{Assumption}}

\newtheorem{theorem}{\bf{Theorem}}

\newtheorem{proof}{Proof}

\ifCLASSINFOpdf
\else
\fi
\usepackage{amsmath}
\usepackage{bm}
\usepackage[cmintegrals]{newtxmath}
\begin{document}
%
\title{Theoretical Analysis of Deep Neural Networks in Physical Layer Communication}
%
%
%

\author{Jun~Liu,~Haitao Zhao,~\IEEEmembership{Senior Member, IEEE},~Dongtang Ma,~\IEEEmembership{Senior Member, IEEE},~Kai~Mei and~Jibo~Wei,~\IEEEmembership{Member, IEEE}
\thanks{Manuscript received February 15,~2022;~revised July 3,~2022;~accepted August 17,~2022.~This work was supported in part by National Natural Science Foundation of China (NSFC) under Grant 61931020, 61372099 and 61601480. This paper has been presented in part at the 2022 IEEE Wireless Communications and Networking Conference Workshops \cite{Liu2022}.~(\emph{Corresponding author:~Jun Liu.})}
\thanks{Jun Liu,~Haitao Zhao,~Dongtang Ma,~Kai Mei and Jibo Wei are with the College of Electronic Science and Technology,~National University of Defense Technology,~Changsha~410073,~China~(E-mail:~\{liujun15,~haitaozhao,~dongtangma,~meikai11,~wjbhw\}@nudt.edu.cn).}

}

%

%
%

\markboth{IEEE Transactions on Communications,~Vol.~XX, No.~X,~August~2022}%
{Shell \MakeLowercase{\textit{et al.}}: Bare Demo of IEEEtran.cls for IEEE Communications Society Journals}
%



\maketitle

\begin{abstract}
Recently, deep neural network (DNN)-based physical layer communication techniques have attracted considerable interest. Although their potential to enhance communication systems and superb performance have been validated by simulation experiments, little attention has been paid to the theoretical analysis. Specifically, most studies in the physical layer have tended to focus on the application of DNN models to wireless communication problems but not to theoretically understand how does a DNN work in a communication system. In this paper, we aim to quantitatively analyze why DNNs can achieve comparable performance in the physical layer comparing with traditional techniques, and also drive their cost in terms of computational complexity. To achieve this goal, we first analyze the encoding performance of a DNN-based transmitter and compare it to a traditional one. And then, we theoretically analyze the performance of DNN-based estimator and compare it with traditional estimators. Third, we investigate and validate how information is flown in a DNN-based communication system under the information theoretic concepts. Our analysis develops a concise way to open the “black box” of DNNs in physical layer communication, which can be applied to support the design of DNN-based intelligent communication techniques and help to provide explainable performance assessment.
\end{abstract}

\begin{IEEEkeywords}
Theoretical analysis, deep neural network (DNN), physical layer communication, information theory.
\end{IEEEkeywords}

%
\IEEEpeerreviewmaketitle

\section{Introduction}
\label{Introduction}
%
%
%
%
\IEEEPARstart{T}{he} mathematical theories of communication systems have been developed dramatically since Claude Elwood Shannon's monograph ``A mathematical theory of communication'' \cite{shannon1948mathematical} provides the foundation of digital communication. However, the wireless channel-related gap between theory and practice still needs to be filled due to the difficulty of precisely modeling wireless channels. Recently, deep neural network (DNN) has drawn a lot of attention as a powerful tool to solve science and engineering problems such as protein structure prediction \cite{senior2020improved}, image recognition \cite{He2016}, speech recognition \cite{AbdelHamid2014} and natural language processing \cite{Hirschberg2015} that are virtually impossible to explicitly formulate. These promising approaches motivate researchers to implement DNNs in existing physical layer communication.

In order to mitigate the gap, a natural thought is to let a DNN to jointly optimize a transmitter and a receiver for a given channel model without being limited to component-wise optimization. Autoencoders (AEs) are considered as a tool to solve this problem. An autoencoder is a type of artificial neural network used to learn efficient codings of unlabeled data. An autoencoder has two main parts: an encoder that maps the input into the code, and a decoder that maps the code to a reconstruction of the input. This structure is equivalent to the concept of a communication system. Along this thread, pure data-driven AE-based end-to-end communication systems are firstly proposed to jointly optimize transmitter and receiver components \cite{o2017introduction,o2018physical,zhu2019joint,MorochoCayamcela2020}. And then, T. O’Shea and J. Hoydis consider the linear and nonlinear steps of processing the received signal as a radio transformer network (RTN) in \cite{o2017introduction}, which can be integrated into the end-to-end training process. The idea of end-to-end learning of communication system and RTN through DNN is extended to orthogonal frequency division multiplexing (OFDM) in \cite{felix2018ofdm} and multiple-input multiple-output (MIMO) in \cite{OShea2017Physical}.

Unfortunately, training these DNNs in practical wireless channels is never a trivial issue. First, these methods require the reliable feedback links. As Shannon once described in \cite{shannon1948mathematical}, the fundamental problem of communication is described as “reproducing at one point either exactly or approximately a message selected at another point”. But, AE-based methods present a “chicken and egg” problem. That is to say, first need a reliable communications system to do the error back-propagation to actually train an end-to-end communication system for you. This leads to the paradox \cite{Liu2021Fine}. To tackle this problem, DNNs are trained offline and then tested online in practical applications. However, this implementation strategy leads to the second problem that these methods assume the availability of the explicit channel model to calculate the gradients for the optimization. Still, the unavailability of perfect channel information forces these methods to adopt simulation-based strategy to train the DNN, which usually results the model mismatch problem. Specifically, the DNN models trained offline show significant performance degradation in test unless both training and test sets are subjected to the same probability distribution. Another resolution is to sample a wireless channel through transmitting probe signal from a neural network-based transmitter. For example, since an AE requires a differentiable channel model, the proposed method in reference \cite{Aoudia2019} calculates the gradient  w.r.t the neural network-based transmitter's parameters through sampling of the channel distribution. Similarly, reference \cite{raj2018backpropagating} utilizes stochastic perturbation technique to train an end-to-end communication system without relying on explicit channel models but the number of training samples is still prohibitive.

Another idea is to estimate the channel as accurate as possible and recover channel state information (CSI) by implementing a DNN at the side of a receiver so that the effects of fading could be reduced. This strategy usually can be divided into two main categories, one using pure data to train a DNN (also known as the \emph{data-driven}) and the other combining data and current models to train a DNN (also known as the \emph{model-driven}). In the data-driven manner, the neural networks (NNs) are optimized merely over a large training data set labeled by true channel responses and a mathematically tractable model is unnecessary \cite{qin2019deep}. The authors of \cite{wen2018deep} propose a data-driven end-to-end DNN-based CSI compression feedback and recovery mechanism which is further extended with long short-term memory (LSTM) to tackle time-varying massive MIMO channels \cite{wang2018deep}. To achieve better estimation performance and reduce computation cost, a compact and flexible deep residual network architecture is proposed to conduct channel estimation for an OFDM system based on downlink pilots in \cite{li2019deep}. Nevertheless, the performance of the data-driven approaches heavily depends on an enormous amount of labeled data which cannot be easily obtained in wireless communication. To address this issue, a plethora of model-driven research gradually has been carried out to achieve efficient receivers \cite{zhang2019deep,Yang2019,Balevi2019,Mei2021,Nguyen2021}. Instead of only using the enormous size of labeled data, in the model-driven manner, domain knowledge is also used to construct a DNN \cite{he2019model}. For example, in model-driven channel estimation, the least-square (LS) estimations usually are fed into a DNN, and then the DNN yields the enhanced channel estimates. Furthermore, in order to mitigate the disturbances, in addition to Gaussian noise, such as channel fading and nonlinear distortion, and further reduce the computation complexity of training, \cite{liu2019online} proposes an online fully complex extreme learning machine (ELM)-based channel estimation and equalization scheme.

Comparing with traditional physical layer communication systems, the above-mentioned DNN-based techniques show competitive performance by simulation experiments. However, the dynamics behind the DNN in physical layer communication remains unknown. In the domain of information theory, a plethora of research has been conducted to investigate the process of learning. In \cite{Tishby2000}, Tishby et al. propose the information bottleneck (IB) method which provides a framework for discussing a variety of problems in signal processing, learning, etc. Then, in \cite{Tishby2015}, DNNs are analyzed via the theoretical framework of the IB principle. In \cite{zaidi2020information}, the variants of the IB problem are discussed. In \cite{Aguerri2021}, Tishby’s centralized IB method is generalized to the distributed setting. Reference \cite{Xu2021} considers the IB problem of a Rayleigh fading MIMO channel with an oblivious relay. However, the considered problems in these researches are different from that in wireless communication standpoint in which the complexity-relevance tradeoffs usually are not considered. Moreover, there are still three major problems. (i) Although it has been shown by simulations that AE-based end-to-end communication systems can approach the optimal symbol error rate (SER), i.e., the SER of a system using optimal constellation, the quantitative comparative analysis has not been conducted. (ii) As a module in a receiver, how does a DNN process information has not been quantitatively investigated. (iii) The methodology to design data sets and the structure of a neural network, which plays an important role in the neural network-based channel estimation, have not been the theoretically studied.

In this paper, we attempt to first give a mathematical explanation to reveal the mechanism of end-to-end DNN-based communication systems. Then, we try to unveil the role of the DNNs in the task of channel estimation. We believe that we have developed a concise way to open as well as understand the ``black box'' of DNNs in physical layer communication, and hence our main contributions of this paper are fourfold:
\begin{itemize}
	\item  Instead of proposing a scheme combining a DNN with a typical communication system, we analyze the behaviors of a DNN-based communication system from the perspectives of the whole DNN (communication system), encoder (transmitter) and decoder (receiver). And we also analyze and compare the performance of the DNN-based transmitter with the conventional method, i.e., the gradient-search algorithm, in terms of both convergence properties and computational complexity.
	
	\item We consider the task of channel estimation as a typical inference problem. With the information theory, we analyze and compare the performance of the DNN-based channel estimation with LS and linear minimum mean-squared error (LMMSE) estimators. Furthermore, we derive the analytical relation between the hyperparameters, as well as training sets, and the performance.
	
	\item We conduct computer simulations and the results verify that the constellations produced by AEs are equivalent to the (locally) optimum constellations obtained by the gradient-search algorithm which minimize the asymptotic probability of error in Gaussian noise under an average power constraint. 
	
	\item Through simulation experiments, our theoretical analysis is validated, and the information flow in the DNNs in the task of channel estimation is estimated by using matrix-based functional of Rényi’s $\alpha$-entropy to approximate Shannon’s entropy.
	
\end{itemize}

To the best of our knowledge, there are typically two approaches to integrate DNN with communication systems. (i) Holistic approach treats a communication system as an end-to-end process, which uses an AE to replace a whole communication system. (ii) Phase-oriented approach, which only investigates the application of DNN in a certain module of communication process \cite{Liao2020}. Therefore, without loss of generality, we mainly investigate two cases: the AE-based communication system and the DNN independently deployed at a receiver.

We note that a shorter conference version of this paper has appeared in IEEE Wireless Communications and Networking Conference (2022). Our initial conference paper gives preliminary simulation results. This manuscript provides comprehensive analysis and proof.

The remainder of this paper is organized as follows. In Section \ref{System Model}, we give the system model and then formulate the problem in Section \ref{Section Encoder} and \ref{Section Decoder} . Next, simulation results are presented in Section \ref{Simulation Results}. Finally, the conclusions are drawn in Section \ref{Conclusion}.

\textit{Notations:}~The notations adopted in the paper are as follows. We use boldface lowercase $\bf{x}$, capital letters $\bf{X}$ and calligraphic letters $\mathcal X$ to denote column vectors, matrices and sets respectively. For a matrix $\bf{X}$, we use ${{\bf{X}}_{ij}}$ to denote its $\left( {i,j} \right)$-th entry. For a vector $\bf{x}$, we use ${\left\| {\bf{x}} \right\|_2}$ to denote the Euclidean norm. For a $m \times n$ matrix $\bf{X}$, we use ${\left\| {\bf{X}} \right\|_F} = \sqrt {\sum\limits_i^m {\sum\limits_j^n {{{\left| {{{\bf{X}}_{ij}}} \right|}^2}} } } $ to denote Frobenius norm and ${\left\| {\bf{X}} \right\|_2} = {\sigma _{\max }}\left( {\bf{X}} \right)$ to denote the operator norm, where ${\sigma _{\max }}\left( {\bf{X}} \right)$ represents the largest singular value of matrix $\bf{X}$. If a matrix $\bf{X}$ is positive semi-definite, we use ${\lambda _{\min }}\left( {\bf{X}} \right)$ to denote its smallest eigenvalue. We use $\left\langle { \cdot , \cdot } \right\rangle $ to denote the standard Euclidean inner product between two vectors or matrices. We let $[n] = \left\{ {1,2, \ldots n} \right\}$. We use ${\mathcal{N}_d}\left( {{\bf{0}},{{\bf{I}}_d}} \right)$ to denote $d$-dimensional standard Gaussian distribution. We also use $O\left(  \cdot  \right)$ to denote standard Big-O only hiding constants. In addition,~$  \odot  $ denotes the Hadamard product,~${\mathbb{E}}\left[  \cdot  \right]$ denotes the expectation operation, and ${\rm{tr}}\left[  \cdot  \right]$ denotes the trace of a matrix.

\begin{figure*}[t]
	\centering
	\includegraphics[width=5.5in]{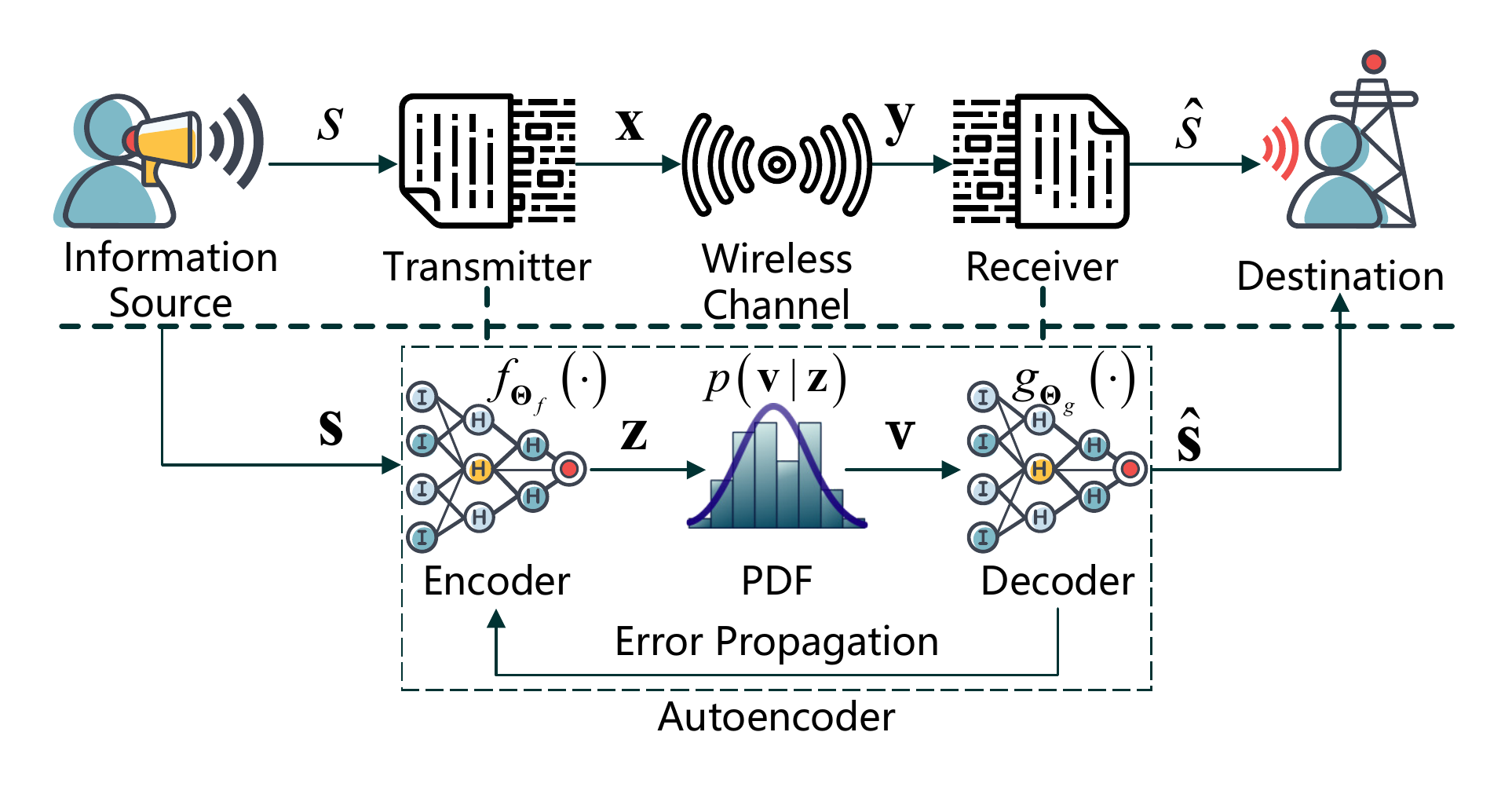}
	\caption{Schematic diagram of a general communication system and its corresponding AE representation.}
	\label{CommunicationsSystems}
\end{figure*}

\section{System Model}
\label{System Model}

In this section, we first describe the considered system model and then provide a detailed explanation of the problem formulation from two different perspectives in the following sections. 
\subsection{Traditional Communication System}
As shown in the upper part of Fig. \ref{CommunicationsSystems}, consider the process of message transmission from the perspective of a typical communication system. We assume that an information source generates a sequence of ${\log _2}M$-bit message symbols ${s} \in \left\{ {1,2, \cdots ,M} \right\}$ to be communicated to the destination. Then the modulation modules inside the transmitter map each symbol $s$ to a signal ${\bf{x}} \in {\mathbb{R}^d}$, where $d$ denoted the dimension of the signal space. The signal alphabet is denoted by ${{\bf{x}}_1},{{\bf{x}}_2}, \cdots ,{{\bf{x}}_M}$. During channel transmission, $d$-dimensional signal $\bf{x}$ is corrupted to ${\bf{y}} \in {\mathbb{R}^d}$ with conditional probability density function (PDF) $p\left( {{\bf{y}}|{\bf{x}}} \right) = \prod _{i = 1}^dp\left( {{y_i}|{x_i}} \right)$. In this paper, we use $d/2$ bandpass channels, each with separately modulated inphase and quadrature components to transmit the $d$-dimensional signal \cite{sklar2014digital}. Finally, the received signal is mapped by the demodulation module inside the receiver to ${\hat s}$ which is an estimate of the transmitted symbol $s$. The procedures mentioned above have been exhaustively presented by Shannon.

\subsection{Understanding Autoencoders on Message Transmission}
From the point of view of filtering or signal inference, the idea of AE-based communication system matches Norbert Wiener's perspective \cite{yu2017autoencoders}. As shown in the lower part of the Fig. \ref{CommunicationsSystems}, the AE consists of an encoder and a decoder and each of them is a feedforward neural network with parameters (weights and biases) ${{{\bf{\Theta }}_f}}$ and ${{{\bf{\Theta }}_g}}$, respectively. Note that each symbol $s$ from information source usually needs to be encoded to a one-hot vector ${\bf{s}} \in {\mathbb{R}^M}$ and then is fed into the encoder. Under a given constraint (e.g., average signal power constraint), the PDF of a wireless channel and a loss function to minimize  symbol error probability, the encoder and decoder are respectively able to learn to appropriately represent $\bf{s}$ as ${\bf{z}} = f\left( {\bf{s}},{{{\bf{\Theta }}_f}} \right)$ and to map the corrupted signal $\bf{v}$ to an estimate of transmitted symbol ${\bf{\hat s}} = g \left( {\bf{v}},{{{\bf{\Theta }}_g}} \right)$ where ${\bf{z}},{\bf{v}} \in {\mathbb{R}^d}$. Here, we use ${{\bf{z}}_1},{{\bf{z}}_2}, \cdots ,{{\bf{z}}_M}$ denoted the transmitted signals from the encoder in order to distinguish it from the transmitted signals from the transmitter.

From the perspective of the whole AE (communication system), it aims to transmit information to a destination with low error probability. The symbol error probability, i.e., the probability that the wireless channel has shifted a signal point into another signal's decision region, is 
\begin{equation}
	{P_e} = \frac{1}{M}\sum\limits_{m = 1}^M {\Pr \left( {{\bf{\hat s}} \ne {{\bf{s}}_m}|{{\bf{s}}_m}~{\rm{transmitted}}} \right)}.
\label{Pe}
\end{equation}
The AE can use the cross-entropy loss function
\begin{equation}
	\begin{aligned}
		{\mathcal L_{\log }}\left( {{\bf{\hat s}},{\bf{s}};{{\bf{\Theta }}_f},{{\bf{\Theta }}_g}} \right) &=  - \frac{1}{n}\sum\limits_{i = 1}^n {\sum\limits_{j = 1}^M {{{\bf{s}}_i}\left[ j \right]\log \left( {{{{\bf{\hat s}}}_i}\left[ j \right]} \right)} }\\ 
		&   =  - \frac{1}{n}\sum\limits_{i = 1}^n {\log \left( {{{{\bf{\hat s}}}_i}\left[ s \right]} \right)} 
	\end{aligned}
\label{LossFunction}
\end{equation}
to represent the cost brought by inaccuracy of prediction where ${{{\bf{s}}_{i}}\left[ j \right]}$ denotes the $j$-th element of the $i$-th symbol in a training set with $n$ symbols. In order to train the AE to minimize the symbol error probability, the optimal parameters could be found by optimizing the loss function
\begin{equation}
	\begin{aligned}
	&\left( {{\bf{\Theta }}_f^ * ,{\bf{\Theta }}_g^ * } \right) = \mathop {\arg \min }\limits_{\left( {{{\bf{\Theta }}_f},{{\bf{\Theta }}_g}} \right)} \left[ {{\mathcal L_{\log }}\left( {{\bf{\hat s}},{\bf{s}};{{\bf{\Theta }}_f},{{\bf{\Theta }}_g}} \right)} \right]\\
		&{\rm{subject~to~}}{\mathbb{E}}\left[ {\left\| {\bf{z}} \right\|_2^2} \right] \le {P_{{\rm{av}}}}
	\end{aligned}
\label{Optimization_AE}
\end{equation}
where $P_{{\rm{av}}}$ denotes the average power. In this paper, we set $P_{{\rm{av}}}={1/M}$. Now, it is important to explain how does the mapping ${\bf{z}} = f \left( {\bf{s}},{{{\bf{\Theta }}_f}} \right)$  vary after the training was done.
\section{Encoder: Finding a Good Representation}
\label{Section Encoder}
Let's pay attention to the encoder (transmitter). In the domain of communication, an encoder needs to learn a robust representation ${\bf{z}} = {f_{{{\bf{\Theta }}_f}}}\left( {\bf{s}} \right)$ to transmit $\bf{s}$ against channel disturbances, including thermal noise, channel fading, nonlinear distortion, phase jitter, etc. This is equivalent to find a coded (or uncoded) modulation scheme with the signal set of size $M$ to map a symbol $\bf{s}$ to a point $\bf{z}$ for a given transmitted power, which maximizes the minimum distance between any two constellation points. Usually the problem of finding good signal constellations for a Gaussian channel\footnote{The problem of constellation optimization is usually considered under the condition of the Gaussian channel. Although the problem under the condition of Rayleigh fading channel has been studied in \cite{boutros1996good}, its prerequisite is that the side information is perfect known.} is associated with the search for lattices with high packing density which is a well-studied problem in the mathematical literature \cite{jorge2015algebraic}. This issue can be addressed through two different methods as follows.

\subsection{Traditional Method: Gradient-Search Algorithm}
The eminent work of \cite{foschini1974optimization} proposed a gradient-search algorithm to obtain the optimum constellations. Consider a zero-mean stationary additive white Gaussian noise (AWGN) channel with one-sided spectral density $2N_0$. For large signal-to-noise ratio (SNR), the asymptotic approximation of (\ref{Pe}) can be written as 
\begin{equation}
	{P_e} \sim \exp \left( { - \frac{1}{{8{N_0}}}\mathop {\min }\limits_{i \ne j} \left\| {{{\bf{z}}_i} - {{\bf{z}}_j}} \right\|_2^2} \right).
\end{equation}
To minimize $P_e$, the problem can be formulated  as
\begin{equation}
	\begin{aligned}
		&\left\{ {{\bf{z}}_m^ * } \right\}_{m = 1}^M = \mathop {\arg \min }\limits_{\left\{ {{{\bf{z}}_m}} \right\}_{m = 1}^M} \left( {{P_e}} \right)\\
		&{\rm{subject~to~}}{\mathbb{E}}\left[ {\left\| {\bf{z}} \right\|_2^2} \right] \le {P_{{\rm{av}}}}
	\end{aligned}
\label{Optimization_Asymptotic}
\end{equation}
where $\left\{ {{\bf{z}}_m^ * } \right\}_{m = 1}^M$ denotes the optimal signal set. This optimization problem can be solved by using a constrained gradient-search algorithm. We denote $\left\{ {{\bf{z}}_m } \right\}_{m = 1}^M$ as an $M \times d$ matrix
\begin{equation}
	{\bf{Z}} = {\left[ {{{\bf{z}}_1},{{\bf{z}}_2}, \cdots ,{{\bf{z}}_M}} \right]^T}.
\end{equation}
Then, the $k$-th step of the constrained gradient-search algorithm can be described by
\begin{subequations}
	\begin{align}
		&{\bf{Z}}_{k + 1}^{\rm{'}} = {{\bf{Z}}_k} - {\eta_k}\nabla {P_e}\left( {{{\bf{Z}}_k}} \right) \\
		&{{\bf{Z}}_{k + 1}} = \frac{{{\bf{Z}}_{k + 1}^{\rm{'}}}}{{\sum\limits_i {\sum\limits_j {{{\left( {{\bf{Z}}_{k + 1}^{\rm{'}}\left[ {i,j} \right]} \right)}^2}} } }}
	\end{align}
\label{MthStep}
\end{subequations}
where $\eta_k$ denotes step size and $\nabla {P_e}\left( {{{\bf{Z}}_k}} \right) \in {\mathbb{R}^{M \times d}}$ denotes the gradient of $P_e$ respect to the current constellation points.
It can be written as
\begin{equation}
	\nabla {P_e}\left( {{{\bf{Z}}_k}} \right) = {\left[ {{{\bf{g}}_1},{{\bf{g}}_2}, \cdots ,{{\bf{g}}_M}} \right]^T}
\end{equation}
where
\begin{equation}
	{{\bf{g}}_m} \sim  - \sum\limits_{i \ne m} {\exp \left( { - \frac{{\left\| {{{\bf{z}}_m} - {{\bf{z}}_i}} \right\|_2^2}}{{8{N_0}}}} \right)\left( {\frac{1}{{\left\| {{{\bf{z}}_m} - {{\bf{z}}_i}} \right\|_2^2}} + \frac{1}{{4{N_0}}}} \right){{\bf{1}}_{{{\bf{z}}_m} - {{\bf{z}}_i}}}}. 
\end{equation}
The vector ${{{\bf{1}}_{{{\bf{z}}_m} - {{\bf{z}}_i}}}}$ denotes $d$-dimensional unit vector in the direction of ${{{\bf{z}}_m} - {{\bf{z}}_i}}$.

Comparing (\ref{Optimization_AE}) to (\ref{Optimization_Asymptotic}), we can understand the mechanism of the encoder in an AE-based communication system. Its optimized variables of AE method are the parameters ${\Theta _f}$ and ${\Theta _g}$. In other words, AE learns the constellation design through simultaneously optimizing the parameters of the encoder and decoder. It does not directly optimize the constellation points ${\bf{s}}$. Differently, the gradient-search algorithm directly optimizes the ${\bf{s}}$. Although the optimized variables of these two methods are different, their optimization goals of minimizing the SER are essentially identical. Therefore, the mapping function of encoder can be represented as
\begin{equation}
\left\{ {f\left( {{{\bf{s}}_m},{\bf{\Theta }}_f^*} \right)} \right\}_{m = 1}^M \to \left\{ {{\bf{z}}_m^*} \right\}_{m = 1}^M
\end{equation}
when the PDF used for generating training samples is multivariate zero-mean normal distribution ${\bf{v}} - {\bf{z}} \sim {{\cal N}_d}({\bf{\vec 0}},{\mkern 1mu} {\bm{\Sigma }})$ where ${{\bf{\vec 0}}}$ denotes $d$-dimensional zero vector and ${\bm{\Sigma }} = \left( {2{N_0}/d} \right){\bf{I}}$ is an $d \times d$ diagonal matrix. In the next subsection, detailed explanation is given.

\subsection{Neural Network-based Method}
Unlike gradient-search algorithm, models based on neural networks are created directly from data by an algorithm. However, under most cases, these models are "black boxes", which means that humans, even those who design them, cannot understand how variables are being combined to make predictions \cite{Rudin2019Why}. At this stage, the task of wireless communication does not require the interpretability of neural networks. However, the theoretical and comparative analyses of a neural network-based communication system are indispensable since an accurate interpretable model has been built.

Let the one-hot vector ${\bf{s}}\in {\mathbb{R}^M}$ be the input ${\bf{x}}$, ${{\bf{W}}^{\left( 1 \right)}} \in {\mathbb{R}^{m \times M}}$ is the first weight matrix, ${{\bf{W}}^{\left( h \right)}} \in {\mathbb{R}^{m \times m}}$ is the weight at the $h$-th layer for $2 \le h \le {H}$ and $\sigma \left(  \cdot  \right)$ is the activation function. We assume intermediate layers are square matrices for the sake of simplicity. There are $H_1$ and $H_2$ hidden layers at the side of transmitter and receiver, respectively. The prediction function can be defined recursively as 
\begin{equation}
	\begin{aligned}
		&{{\bf{x}}^{\left( h \right)}} = \sqrt {\frac{{{c_\sigma }}}{m}} \sigma \left( {{{\bf{W}}^{\left( h \right)}}{{\bf{x}}^{\left( {h - 1} \right)}}} \right),\\
		&h \in \left[{H} \right]\setminus \left\{ {{H_1} + 1} \right\}\
	\end{aligned}
\end{equation}
where ${c_\sigma } = {\left( {{\mathbb{E}_{x \sim \mathcal{N}\left( {0,1} \right)}}\left[ {\sigma {{\left( x \right)}^2}} \right]} \right)^{ - 1}}$ is a scaling factor to normalize the input in the initialization phase, and $\left(H_1+1\right)$-th layer is defined as channel layer. Note that $\left[{H} \right]\setminus \left\{ {{H_1} + 1} \right\}$ is the set of elements that belong to $\left[{H} \right]$ but not to $\left\{ {{H_1} + 1} \right\}$. To constrain the average power of transmitted signal $P_{{\rm{av}}}$ to ${1/M}$, ${\bf{x}}^{H_1}$ is normalized as 	  
\begin{equation}
	{\bf{z}} = f\left( {{\bf{x}},{\Theta _f}} \right) = \sqrt {\frac{1}{{\mathbb{E}\left[ {\left\| {{{\bf{x}}^{\left( {{H_1}} \right)}}} \right\|_2^2} \right]M}}} {{\bf{x}}^{\left( {{H_1}} \right)}}.
\end{equation}
Then, the effect on the transmitted signal resulting from a wireless channel can be expressed as  
\begin{equation}
	{\bf{v}} = h\left( {{\bf{z}},{\Theta _h}} \right)
	\label{Channel Layer}
\end{equation}
where ${h_{{\Theta _h}}}$ is the functional form of the wireless channel with parameter set $\Theta _h$\footnote{In accordance with practice, and without introducing ambiguity, $h$ is used to denote hidden layer and channel in the context of neural network and physical communication, respectively.}. Let ${{\bf{x}}^{\left( H_1+1 \right)}}=\bf{v}$, and finally, the received signal is demodulated as 
\begin{equation}
{\bf{\hat s}} = g\left( {{{\bf{x}}^{\left( {{H_1} + 1} \right)}},{\Theta _g}} \right) = {\sigma ^{\left( H \right)}}\left( {{{\bf{W}}^{\left( H \right)}}{{\bf{x}}^{\left( {H - 1} \right)}}} \right)
\end{equation}
where ${\sigma ^{\left( H \right)}}\left(  \cdot  \right) = {\rm{softmax}}\left(  \cdot  \right)$.

To theoretically analyze the neural network, some technical conditions on the activation function are imposed. The first condition is \emph{Lipschitz and Smooth}. The second is that $\sigma \left(  \cdot  \right)$ is analytic and is not a polynomial function. In this section, softplus ${\sigma ^{\left( h \right)}}\left( z \right) = \log \left( {1 + \exp \left( z \right)} \right)$ is chosen for the hidden layers except for the ${H_1} + 1$ and $H$-th layers. 

While training the deep neural network, randomly initialized gradient descent algorithm is used to optimize the empirical loss (\ref{LossFunction}). Specifically, for every layer $h \in \left[ {H} \right]\backslash \left\{ {{H_1} + 1} \right\}$, each entry is sample from a standard Gaussian distribution, ${\bf{W}}_{ij}^{\left( h \right)} \sim \mathcal{N}\left( {0,1} \right)$. Then, the values for parameters can be updated by gradient descent, for $k = 1,2, \ldots ,$ and $h \in \left[ {H} \right]\backslash \left\{ {{H_1} + 1} \right\}$ as 
\begin{equation}
	{{\bf{W}}^{\left( h \right)}}\left( k \right) = {{\bf{W}}^{\left( h \right)}}\left( {k - 1} \right) - \eta \frac{{\partial \mathcal{L}_{\rm{log}}\left( {\Theta \left( {k - 1} \right)} \right)}}{{\partial {{\bf{W}}^{\left( h \right)}}\left( {k - 1} \right)}}
\end{equation}
where $\eta  > 0$ is the step size.

The update of parameters ${\Theta _g}$ at the side of the receiver can be realized as
\begin{equation}
	{\Theta _g}\left( k \right) = {\Theta _g}\left( {k - 1} \right) - \eta \frac{{\partial {\mathcal L_{\log }}\left( {\Theta \left( {k - 1} \right)} \right)}}{{\partial {g_{{\Theta _g}}}}}\frac{{\partial {g_{{\Theta _g}}}}}{{\partial {\Theta _g}}}
\end{equation}
since we know the function ${g_{{\Theta _g}}}\left(  \cdot  \right)$ entirely. At the side of the transmitter, it becomes 
\begin{equation}
		{\Theta _f}\left( k \right) = {\Theta _f}\left( {k - 1} \right) - \eta \frac{{\partial {\mathcal L_{\log }}\left( {\Theta \left( {k - 1} \right)} \right)}}{{\partial {g_{{\Theta _g}}}}}\frac{{\partial {g_{{\Theta _g}}}}}{{\partial {h_{{\Theta _h}}}}}\frac{{\partial {h_{{\Theta _h}}}}}{{\partial {f_{{\Theta _f}}}}}\frac{{\partial {f_{{\Theta _f}}}}}{{\partial {\Theta _f}}}
\label{Backpropagation to Transmitter}
\end{equation}
where the terms $\frac{{\partial {g_{{\Theta _g}}}}}{{\partial {h_{{\Theta _h}}}}}$ and $\frac{{\partial {h_{{\Theta _h}}}}}{{\partial {f_{{\Theta _f}}}}}$ are difficult to acquire unless the knowledge about the channel ${{h_{{\Theta _h}}}}$ is fully known. In this paper, we consider both the Gaussian channel and the Rayleigh flat fading channel.

\subsubsection{Gaussian Channel}
Let ${\bf{n}} \in {\mathbb{R}^m}$ is a white Gaussian noise vector and the variance of each entry is $\sigma _n^2$. The output of the channel layer can be expressed as
\begin{equation}
	{{\bf{x}}^{\left( {{H_1} + 1} \right)}} = {{\bf{W}}^{\left( {{H_1} + 1} \right)}}c_\sigma ^{\left( {{H_1}} \right)}{{\bf{x}}^{\left( {{H_1}} \right)}} + {\bf{n}}
	\label{Gaussian Channel}
\end{equation}
where  ${{\bf{W}}^{\left( {{H_1} + 1} \right)}} = {{\bf{I}}_m}$ and $c_\sigma ^{\left( {{H_1}} \right)}=\sqrt {1/\left( {\mathbb{E}\left[ {\left\| {{{\bf{x}}^{\left( {{H_1}} \right)}}} \right\|_2^2} \right]M} \right)} $. Then, (\ref{Gaussian Channel}) can be expressed as
\begin{equation}
	{{\bf{x}}^{\left( {{H_1} + 1} \right)}} = {{{\bf{W'}}}^{\left( {{H_1} + 1} \right)}}{{{\bf{x'}}}^{\left( {{H_1}} \right)}}
	\label{Equivalence Weight}
\end{equation}
where ${{{\bf{W'}}}^{\left( {{H_1} + 1} \right)}} = \left[ {{{\bf{I}}_m},{\bf{n}}} \right] \in {{\mathbb{R}}^{m \times \left( {m + 1} \right)}}$ denotes the equivalent weights of the channel layer and ${{{\bf{x'}}}^{\left( {{H_1}} \right)}} = \left[ {c_\sigma ^{\left( {{H_1}} \right)}{{{\bf{x'}}}^{\left( {{H_1}} \right)}};1} \right] \in {\mathbb{R}^{m + 1}}$. Finally, the terms $\frac{{\partial {g_{{\Theta _g}}}}}{{\partial {h_{{\Theta _h}}}}}$ and $\frac{{\partial {h_{{\Theta _h}}}}}{{\partial {f_{{\Theta _f}}}}}$ can be written as
\begin{subequations}
	\begin{align}
		&\frac{{\partial {g_{{\Theta _g}}}}}{{\partial {h_{{\Theta _h}}}}}=\frac{{\partial {g_{{\Theta _g}}}}}{{\partial {{{\bf{W'}}}^{\left( {H_1 + 1} \right)}}}}\\
		&\frac{{\partial {h_{{\Theta _h}}}}}{{\partial {f_{{\Theta _f}}}}} = \frac{{\partial {{{\bf{W'}}}^{\left( {{H_1} + 1} \right)}}}}{{\partial {f_{{\Theta _f}}}}}.
	\end{align}	
\label{Gradient of Channel_Gaussian}
\end{subequations}
Substituting \ref{Gradient of Channel_Gaussian} into (\ref{Backpropagation to Transmitter}), we get
\begin{equation}
	\begin{aligned}
		{\Theta _f}\left( k \right) =& {\Theta _f}\left( {k - 1} \right) -\\
		&\eta \frac{{\partial {L_{\log }}\left( {\Theta \left( {k - 1} \right)} \right)}}{{\partial {g_{{\Theta _g}}}}}\frac{{\partial {g_{{\Theta _g}}}}}{{\partial {{{\bf{W'}}}^{\left( {{H_1} + 1} \right)}}}}\frac{{\partial {{{\bf{W'}}}^{\left( {{H_1} + 1} \right)}}}}{{\partial {f_{{\Theta _f}}}}}\frac{{\partial {f_{{\Theta _f}}}}}{{\partial {\Theta _f}}}.
	\end{aligned}
	\label{Analysis of Backpropagation to Transmitter_Gaussian}
\end{equation}
\subsubsection{Fading Channel}
We consider a Rayleigh flat fading channel for simplicity. It is not difficult to generalize our analysis to other fading channels, e.g., frequency selective fading channels. 

To transmitted the modulated signal ${{\bf{x'}}^{\left( {{H_1}} \right)}}$, $m/2$ bandpass channels are needed. We assume that the channel impulse of each channel is mutually independent, i.e., ${{{\bf{W'}}}^{\left( {{H_1} + 1} \right)}} = \left[ {{\bf{H}},{\bf{n}}} \right] \in {^{m \times \left( {m + 1} \right)}}$ where 
\begin{equation}
	{\bf{H}} = \left[ {\begin{array}{*{20}{c}}
			{{h_{1,{\rm{I}}}}}&0& \cdots &0\\
			0&{{h_{1,{\rm{Q}}}}}& \cdots &0\\
			\vdots & \vdots & \ddots & \vdots \\
			0&0& \cdots &{{h_{m/2,{\rm{Q}}}}}
	\end{array}} \right],
\end{equation}
and $\left( {{h_{1,{\rm{I}}}},{h_{1,{\rm{Q}}}}, \ldots ,{h_{m/2,{\rm{Q}}}}} \right) \sim {{\mathcal{N}}_m}\left( {0,{{\bf{I}}_m}} \right)$. The real and imaginary parts of the channel impulse of the $i$-th bandpass channel are denoted as ${h_{i,{\rm{I}}}}$ and ${h_{i,{\rm{Q}}}}$, respectively. In this case, the parameters update of parameters at this side of transmitter has the same form as (\ref{Analysis of Backpropagation to Transmitter_Gaussian}).

\subsection{Properties of Convergence}
The convergence properties of the traditional algorithm, i.e., gradient-search algorithm, have been exhaustively analyzed in \cite{foschini1974optimization}. Whereas, the properties of the AE-abased algorithm have not been studied yet, and therefore, we mainly try to analyze the convergence properties of the AE-based algorithm in this subsection. 

In \cite{du2018gradient}, Simon S. Du et al. analyze two-layer fully connected ReLU activated neural networks. It has been shown that, with over-parameterization, gradient descent provably converges to the global minimum of the empirical loss at a linear convergence rate. Then, they develop a unified proof strategy for the fully-connected neural network, ResNet and convolutional ResNet \cite{Du2019}.

We replace (\ref{LossFunction}) with square loss function as 
\begin{equation}
	{\mathcal{L}_2} = \frac{1}{2}\sum\limits_{i = 1}^n {{{\left( {{{{\bf{\hat s}}}_i}\left[ s \right] - 1} \right)}^2}} 
\end{equation}
Then, the individual prediction at the $k$-th iteration can be denoted as
\begin{equation}
	{{\hat s}_i} = {g_{{\Theta _g}}}\left( {{f_{{\Theta _f}}}\left( {{{\bf{s}}_i}} \right)} \right)\left[ s \right].
\end{equation} 
Let ${\bf{\hat s}}\left( k \right) = {\left( {{{\hat s}_1}\left( k \right), \ldots ,{{\hat s}_n}\left( k \right)} \right)^T} \in {\mathbb{R}^n}$ where $n$ denotes the size of the training set. Simon S. Du et al. \cite{du2018gradient, Du2019} show that for DNN, the sequence $\left\{ {{\bf{1}} - {\bf{\hat s}}\left( k \right)} \right\}_{k = 0}^\infty $ admits the dynamics
\begin{equation}
	{\bf{1}} - {\bf{\hat s}}\left( {k + 1} \right) = \left( {{\bf{I}} - \eta {\bf{G}}\left( k \right)} \right)\left( {{\bf{1}} - {\bf{\hat s}}\left( k \right)} \right)
	\label{Linear Dynamics}
\end{equation}
where
\begin{equation}
	\begin{aligned}
		{{\bf{G}}_{ij}}\left( k \right) &= \left\langle {\frac{{\partial {{\hat s}_i}\left( k \right)}}{{\partial \Theta \left( k \right)}},\frac{{\partial {{\hat s}_j}\left( k \right)}}{{\partial \Theta \left( k \right)}}} \right\rangle\\
		& = \sum\limits_{h = 1}^H {\left\langle {\frac{{\partial {{\hat s}_i}\left( k \right)}}{{\partial {{\bf{W}}^{\left( h \right)}}\left( k \right)}},\frac{{\partial {{\hat s}_j}\left( k \right)}}{{\partial {{\bf{W}}^{\left( h \right)}}\left( k \right)}}} \right\rangle }\\
		& \buildrel \Delta \over = \sum\limits_{h = 1}^H {{\bf{G}}_{ij}^{\left( h \right)}\left( k \right)}.  
	\end{aligned}
\end{equation}

The Gram matrix induced by the weights from $h$-th layer ${{\bf{G}}^{\left( h \right)}} \in {\mathbb{R}^{n \times n}}$ is defined as ${\bf{G}}_{ij}^{\left( h \right)}\left( k \right) = \left\langle {\frac{{\partial {{\hat s}_i}\left( k \right)}}{{\partial {{\bf{W}}^{\left( h \right)}}\left( k \right)}},\frac{{\partial {{\hat s}_j}\left( k \right)}}{{\partial {{\bf{W}}^{\left( h \right)}}\left( k \right)}}} \right\rangle $ for $h = 1, \ldots ,H$. Note for all $h \in \left[ H \right]$, each entry of ${{\bf{G}}^{\left( h \right)}}$ is an inner product.

In \cite{Du2019}, it has been shown that, if the width is large enough for all layers, for all $k = 0,1, \ldots $, ${{\bf{G}}^H}\left( k \right)$ is close to a fixed matrix ${{\bf{K}}^{\left( H \right)}} \in {\mathbb{R}^{n \times n}}$ which depends on the input data, neural network architecture and the activation but does not depend on neural network parameters $\Theta $. The Gram matrix ${{\bf{K}}^{\left( H \right)}}$ is recursively defined as follows. For $\left( {i,j} \right) \in \left[ n \right] \times \left[ n \right]$ and $h \in \left[ {H - 1} \right]$,
\begin{subequations}
	\begin{align}
		&{\bf{K}}_{ij}^{\left( 0 \right)} = \left\langle {{{\bf{x}}_i},{{\bf{x}}_j}} \right\rangle ,\\
		&{\bf{A}}_{ij}^{\left( h \right)} = \left( {\begin{array}{*{20}{c}}
				{{\bf{K}}_{ii}^{\left( {h - 1} \right)}}&{{\bf{K}}_{ij}^{\left( {h - 1} \right)}}\\
				{{\bf{K}}_{ji}^{\left( {h - 1} \right)}}&{{\bf{K}}_{jj}^{\left( {h - 1} \right)}}
		\end{array}} \right),\\
		&{\bf{K}}_{ij}^{\left( h \right)} = {c_\sigma }{\mathbb{E}_{{{\left( {u,v} \right)}^T}}} \sim \mathcal{N}\left( {{\bf{0}},{\bf{A}}_{ij}^{\left( h \right)}} \right)\left[ {\sigma \left( u \right)\sigma \left( v \right)} \right],\\
		&{\bf{K}}_{ij}^{\left( H \right)} = {c_\sigma }{\bf{K}}_{ij}^{\left( {H - 1} \right)}{\mathbb{E}_{{{\left( {u,v} \right)}^T}}} \sim \mathcal{N}\left( {{\bf{0}},{\bf{A}}_{ij}^{\left( h \right)}} \right)\left[ {\sigma '\left( u \right)\sigma '\left( v \right)} \right].
	\end{align}	
	\label{Gram Matrix K}
\end{subequations}

\begin{figure*}[t]
	\centering
	\includegraphics[width=5.5in]{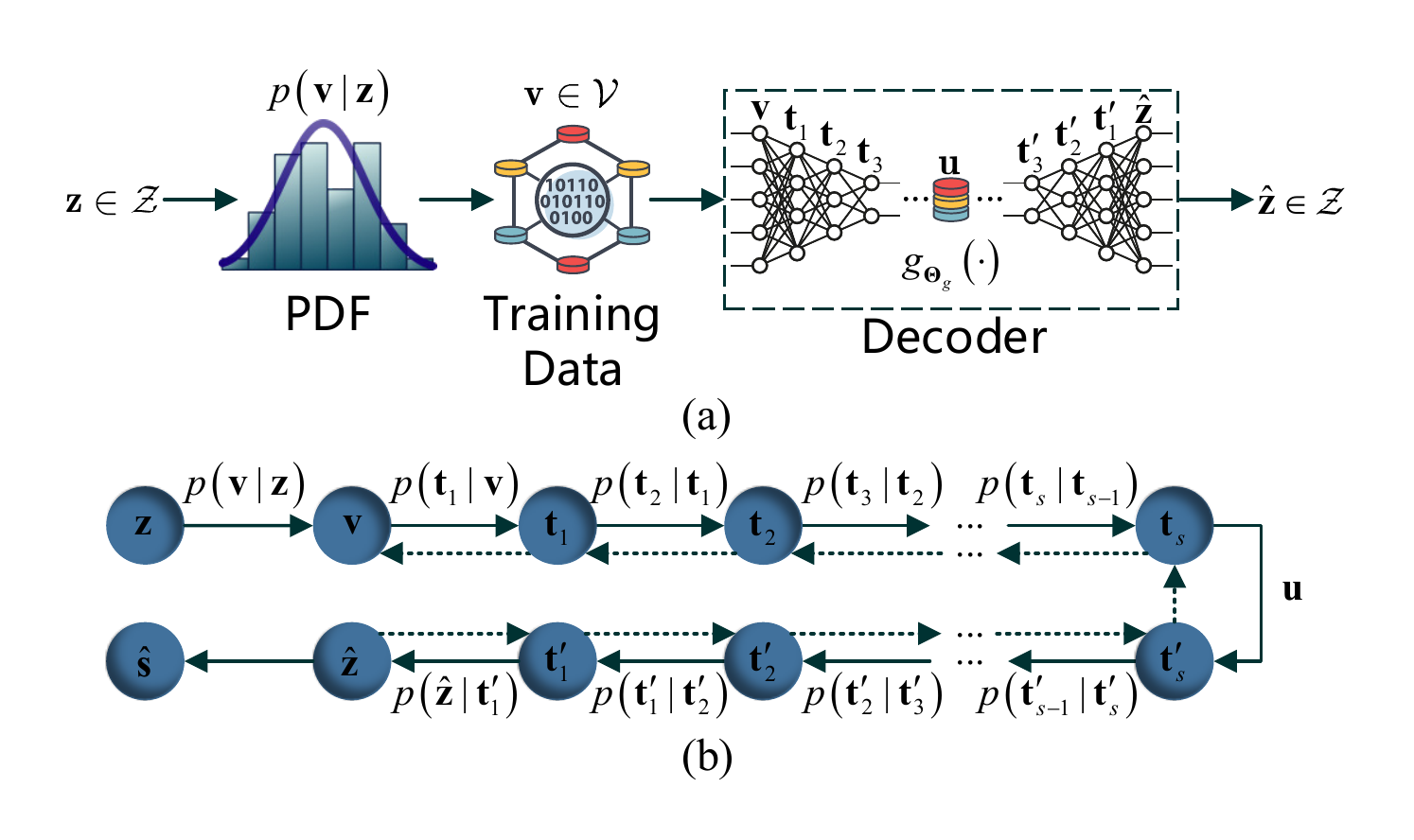}
	\caption{(a) An inference model with a DNN decoder of size $\left( {2S - 1} \right)$ hidden layers for learning. (b) The graph representation of the decoder with $\left( {S - 1} \right)$ hidden layers in both sub encoder and decoder. The solid arrow denotes the direction of input feedforward propagation and the dashed arrow denotes the direction of information flow in the error back-propagation phase.}
	\label{Decoder}
\end{figure*}

However, the existence of the channel layer obstructs this recursive process. Specifically, under the case of Gaussian channel, ${{\bf{W'}}^{\left( {{H_1} + 1} \right)}} = \left[ {{{\bf{I}}_m},{\bf{n}}} \right]$ results that the strictly positive definiteness of matrix ${{\bf{K}}^{\left( H \right)}}$ may not be guaranteed. This leads that the dynamics of gradient descent does not enjoy a linear convergence rate as (\ref{Linear Dynamics}) shown.

Under the case of Rayleigh flat fading channel, ${{\bf{W'}}^{\left( {{H_1} + 1} \right)}} = \left[ {{\bf{H}},{\bf{n}}} \right]$ makes the situation worse. Since the fading channel is not static, the diagonal elements of the weights matrix ${\bf{H}}\left( k \right)$ in the ${{\bf{W'}}^{\left( {{H_1} + 1} \right)}}$ at $k$-th iterative are a random sample from ${{\cal N}_m}\left( {0,{{\bf{I}}_m}} \right)$. At the stage of random initialization, suppose we have some perturbation ${\left\| {{{\bf{G}}^{\left( 1 \right)}}\left( 0 \right) - {{\bf{K}}^{\left( 1 \right)}}} \right\|_2} \le {\varepsilon _1}$ in the first layer. This perturbation propagates to the $H$-th layer admits the form
\begin{equation}
	{\left\| {{{\bf{G}}^{\left( H \right)}}\left( 0 \right) - {{\bf{K}}^{\left( H \right)}}} \right\|_2} \buildrel \Delta \over = {\varepsilon _H}\mathop  < \limits_{\sim} {2^{O\left( H \right)}}{\varepsilon _1}.
\end{equation}

Therefore, we need to have ${\varepsilon _1} \le 1/{2^{O\left( H \right)}}$ and this makes $m$ have exponential dependency on $H$  \cite{Du2019}. Moreover, at the training stage, the perturbation in the $\left( {{H_1}{\rm{ + 1}}} \right)$-th layer induced by fading channel can disperse to the whole network, i,e., the averaged Frobenius norm 
\begin{equation}
	\frac{1}{{\sqrt m }}{\left\| {{{\bf{W}}^{\left( h \right)}}\left( k \right) - {{\bf{W}}^{\left( h \right)}}\left( 0 \right)} \right\|_F}
\end{equation} 
is not small for all $k = 0,1, \ldots $.

In addition, large biases $\bf{n}$ would be introduced by the channel noise when SNR is low. Although the biases do not impact the weights directly, its influence can be spread to the whole network through forward and backward propagation.

\section{Decoder: Inference}
\label{Section Decoder}
In this section, we will zoom in the lower right corner of the Fig. \ref{CommunicationsSystems} to investigate what happens inside the decoder (receiver). As Fig. \ref{Decoder}(a) shows, for the task of DNN-based channel estimate, the problem can be formulated as an inference model. For the sake of convenience, we can denote the target output of the decoder as $\bf{z}$ instead of $\bf{s}$ because we can assume  ${\bf{z}} = {f_{{{\bf{\Theta }}_f}}}\left( {\bf{s}} \right)$ is bijection. If the decoder is symmetric, the decoder also can be seen as a sub AE which consists of a sub encoder and decoder. Its bottleneck (or middlemost) layer codes is denoted as $\bf{u}$. Here we use $\bf{z}$ to denote CSI or transmitted symbol which we desire to predict. The decoder infers a prediction ${\bf{\hat z}} = {g_{{{\bf{\Theta }}_g}}}\left( {\bf{v}} \right)$ according to its corresponding measurable variable $\bf{v}$. This inference problem can be addressed by the following two different methods.
\subsection{Traditional Estimation Method}
\subsubsection{LS Estimator}
Without loss of generality, we can consider this issue under the context of complex channel estimation. Let the measurable variable be 
\begin{equation}
	{\bf{v}} = {{{\bf{\hat h}}}_{{\rm{LS}}}} = {\bf{h}} + {\bf{n}}
	\label{LS Model}
\end{equation}
where ${{{\bf{\hat h}}}_{{\rm{LS}}}}$ denotes the LS estimation of its corresponding true channel responses ${\bf{h}} \in {\mathbb{C}^d}$, and ${\bf{n}} \in {\mathbb{C}^d}$ is a vector of \textit{i.i.d.} complex zero-mean Gaussian noise with variance $\sigma _n^2$. The noise $\bf{n}$ is assumed to be uncorrelated with the channel $\bf{h}$. The corresponding MSE is 
\begin{equation}
	{\rm{MS}}{{\rm{E}}_{{\rm{LS}}}} = {\mathbb{E}}\left[ {\left\| {{\bf{h}} - {{\bf{h}}_{{\rm{LS}}}}} \right\|_2^2} \right] = d\sigma _n^2.
	\label{MSE_LS}
\end{equation}
\subsubsection{LMMSE Estimator}
If the channel is stationary and subject to ${\bf{h}} \sim {\mathcal{N}_{\mathcal{C}}}\left( {{\bf{0}},{{\bf{R}}_{{\bf{hh}}}}} \right)$, its LMMSE estimate can be expressed as 
\begin{equation}
	{{{\bf{\hat h}}}_{{\rm{LMMSE}}}} = {{\bf{R}}_{{\bf{hh}}}}{\left( {{{\bf{R}}_{{\bf{hh}}}} + \sigma _n^2{{\left( {{\bf{X}}{{\bf{X}}^H}} \right)}^{ - 1}}} \right)^{ - 1}}{{{\bf{\hat h}}}_{{\rm{LS}}}}
	\label{LMMSE Channel Estimation}
\end{equation}
where ${{\bf{R}}_{{\bf{hh}}}} = E\left\{ {{\bf{h}}{{\bf{h}}^H}} \right\}$ is the channel autocorrelation matrix and ${\bf{X}}$ is a diagonal matrix containing the known transmitted signaling points \cite{Beek1995,Edfors1998,Mei2021a}. The MSE of the LMMSE  estimate is 
\begin{equation}
	\begin{aligned}
		{\rm{MS}}{{\rm{E}}_{{\rm{LMMSE}}}} &= \mathbb{E}\left[ {\left\| {{\bf{h}} - {{\bf{h}}_{{\rm{LMMSE}}}}} \right\|_2^2} \right]\\
		&={\rm{tr}}\left[ {{{\bf{R}}_{{\bf{hh}}}}{{\left( {{{\bf{I}}_d} + \frac{1}{{\sigma _n^2}}{{\bf{R}}_{{\bf{hh}}}}} \right)}^{ - 1}}} \right] \le {\rm{MS}}{{\rm{E}}_{{\rm{LS}}}}.
	\end{aligned}
	\label{MSE_LMMSE}
\end{equation}
To perform (\ref{LMMSE Channel Estimation}), ${{{\bf{R}}_{{\bf{hh}}}}}$ and ${\sigma _n^2}$ are assumed to be known.
\begin{figure}[t]
	\centering
	\includegraphics[width=3.50in]{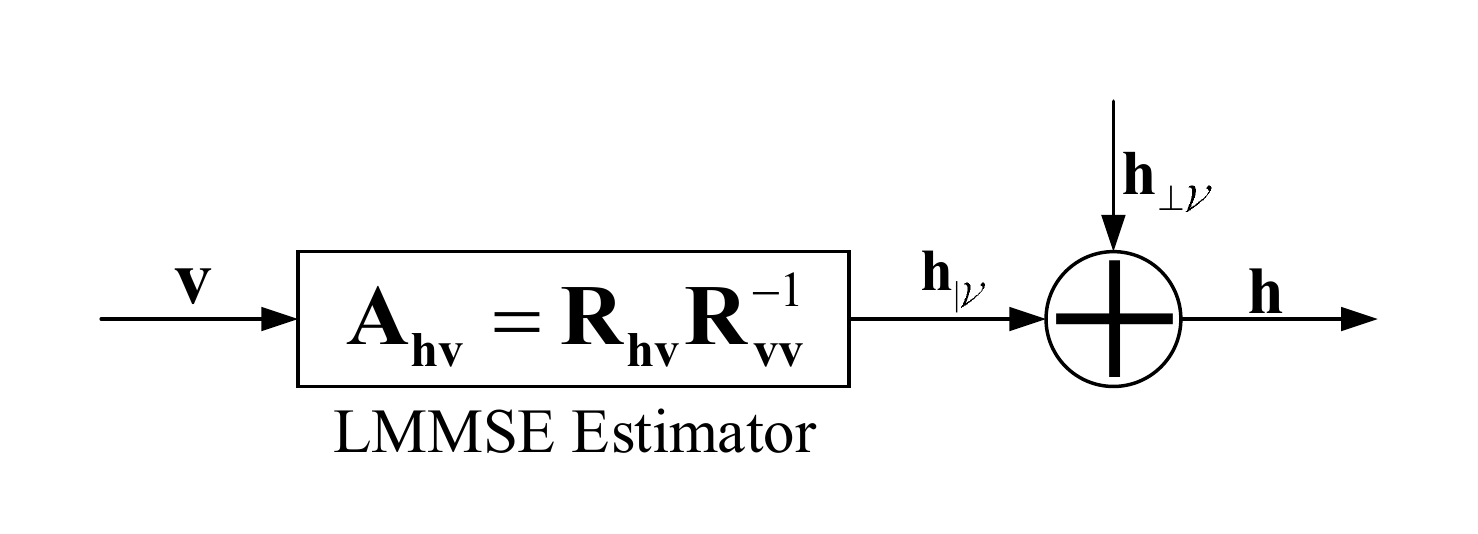}
	\caption{The block diagram of the LMMSE estimator.}
	\label{LMMSE Estimator}
\end{figure}
We define a complex space $\cal{G}$. Every element of $\cal{G}$ is a finite-variance, zero-mean, proper complex Gaussian random variable, and every subset of $\cal{G}$ is jointly Gaussian. The set of observed variables and its closure or the subspace generated by ${\cal{V}}$ are denoted as $\cal{V}\in\cal{G}$ and ${\bar{\cal{V}}}$, respectively. Let  ${\bf{e}} = {\bf{h}} - {\bf{\hat h}}$ be the estimation error where ${\bf{\hat h}} \in \bar {\cal{V}}$ is a linear estimate of $\bf{h}$. By the projection and Pythagorean theorems, we have 
\begin{equation}
	\begin{aligned}
		{\left\| {\bf{e}} \right\|^2} &= {\left\| {{{\bf{h}}_{|\cal V}} + {{\bf{h}}_{ \bot \cal V}} - {\bf{\hat h}}} \right\|^2} = {\left\| {{{\bf{h}}_{|\cal V}} - {\bf{\hat h}}} \right\|^2} + {\left\| {{{\bf{h}}_{ \bot \cal V}}} \right\|^2}\\
		& \ge {\left\| {{{\bf{h}}_{ \bot \cal V}}} \right\|^2}
	\end{aligned},
\end{equation}
with equality if and only if ${\bf{\hat h}} = {{\bf{h}}_{|\cal V}}$. For this reason, ${{{\bf{h}}_{|\cal V}}}$ is called the LMMSE estimate of $\bf{h}$ given $\bf{v}$, and ${{{\bf{h}}_{ \bot \cal V}}}$ is called the MMSE estimation error. Moreover, the orthogonality principle holds: ${\bf{\hat h}} \in \bar {\cal V}$ is the LMMSE estimate of $\bf{h}$ given $\bf{v}$ if and only if ${\bf{h}} - {\bf{\hat h}}$ is orthogonal to $\bar {\cal {V}}$. Writing ${{{\bf{\hat h}}}_{{\rm{LMMSE}}}}$ as a set of linear combinations of the elements of $\bf{v}$ in matrix form, namely ${{{\bf{\hat h}}}_{{\rm{LMMSE}}}} = {{\bf{A}}_{{\bf{hv}}}}{\bf{v}}$, we obtain a unique solution ${{\bf{A}}_{{\bf{hv}}}} = {{\bf{R}}_{{\bf{hv}}}}{\bf{R}}_{{\bf{vv}}}^{ - 1}$. Fig. \ref{LMMSE Estimator} illustrates that $\bf{h}$ can be decomposed into a linear estimate derived from $\bf{v}$ and an independent error variable ${\bf{e}} = {{\bf{h}}_{ \bot {\cal V}}}$. Moreover, this block diagram implies that the MMSE estimate ${{\bf{h}}_{|\cal V}}$ is a sufficient statistic for estimation of $\bf{h}$ from $\bf{v}$, since ${\bf{v}} - {{\bf{h}}_{|v}} - {\bf{h}}$ is evidently a Markov chain; i.e., $\bf{v}$ and $\bf{h}$ are conditionally independent given ${{\bf{h}}_{|\cal V}}$. This is also known as the sufficiency property of the MMSE estimate \cite{forney2004shannon}.

By the sufficiency property, the MMSE estimate ${{\bf{h}}_{|\cal V}}$ is a function of $\bf{v}$ that satisfies the data processing inequality of information theory with equality: $I\left( {{\bf{h}};{\bf{v}}} \right) = I\left( {{\bf{h}};{{\bf{h}}_{|\cal V}}} \right)$ \footnote{Note that no confusion should arise if we abuse the notation slightly by using a lower-case letter to denote a random variable.}. In other words, the reduction of $\bf{v}$ to ${{\bf{h}}_{|\cal V}}$ is information-lossless. Since ${\bf{h}} = {{\bf{A}}_{{\bf{hv}}}}{\bf{v}} + {{\bf{h}}_{ \bot {\cal V}}}$ is a linear Gaussian channel model with Gaussian input $\bf{v}$, Gaussian output $\bf{h}$, and independent additive Gaussian noise ${\bf{e}} = {{\bf{h}}_{ \bot {\cal V}}}$, we have
\begin{equation}
	I\left( {{\bf{h}};{\bf{v}}} \right) = h\left( {\bf{h}} \right) - h\left( {{\bf{h}}|{\bf{v}}} \right) = h\left( {\bf{h}} \right) - h\left( {\bf{e}} \right) = \log \frac{{\left| {{{\bf{R}}_{{\bf{hh}}}}} \right|}}{{\left| {{{\bf{R}}_{{\bf{ee}}}}} \right|}},
	\label{Information Theory_MMSE}
\end{equation}
where
\begin{equation}
	\begin{aligned}
		{{\bf{R}}_{{\bf{ee}}}} &= \mathbb{E}\left[ {\left( {{\bf{h}} - {{{\bf{\hat h}}}_{{\bf{LMMSE}}}}} \right){{\left( {{\bf{h}} - {{{\bf{\hat h}}}_{{\bf{LMMSE}}}}} \right)}^H}} \right]\\
		& = \mathbb{E}\left[ {\left( {{\bf{h}} - {{{\bf{\hat h}}}_{{\bf{LMMSE}}}}} \right){{\bf{h}}^H}} \right] + \mathbb{E}\left[ {\left( {{\bf{h}} - {{{\bf{\hat h}}}_{{\bf{LMMSE}}}}} \right){{\bf{h}}^H}{\bf{A}}_{{\bf{hv}}}^H} \right]\\
		& = \mathbb{E}\left[ {\left( {{\bf{h}} - {{{\bf{\hat h}}}_{{\bf{LMMSE}}}}} \right){{\bf{h}}^H}} \right]\\
		&={{\bf{R}}_{{\bf{hh}}}} - {{\bf{A}}_{{\bf{hv}}}}{{\bf{R}}_{{\bf{vh}}}}.
	\end{aligned}
\label{Error Covariance Matrix}
\end{equation}

\subsection{DNN-based Estimation}
Let ${\cal H}_{k,{l}}^\sigma $ be the hypothesis class associated with a $l$-layer neural network of size $k$ with activation function $\sigma$. We assume that $\left| {\sigma \left( z \right)} \right| \le 1$ for all $z \in \mathbb R$. More specifically, ${\cal H}_{k,{l}}^\sigma $ is the set of all functions $h:{{\mathbb R}^d} \to {{\mathbb R}^d}$. Given a random sample $D_n= \left\{ {\left( {{{\bf{v}}_1},{{\bf{z}}_1}} \right), \ldots ,\left( {{{\bf{v}}_n},{{\bf{z}}_n}} \right)} \right\}$, we define 
\begin{equation}
	{\rm{MMS}}{{\rm{E}}_{k,l,n}}\left( {{\bf{z}}|{\bf{v}}} \right): = \mathop {\inf }\limits_{h \in H_{k,l}^\sigma } \frac{1}{n}\sum\limits_{i = 1}^n {\left\| {{{\bf{z}}_i} - h\left( {{{\bf{v}}_i}} \right)} \right\|_2^2},
\end{equation}
i.e., ${\rm{MMS}}{{\rm{E}}_{k,l,n}}\left( {{\bf{z}}|{\bf{v}}} \right)$ is the minimum empirical square loss attained by a $l$-layer neural network of size $k$. Mario Diaz \emph{et al.} establishes a probabilistic bound for the MMSE in estimating ${\bf{z}} \in {\mathbb R^1}$ given ${\bf{v}} \in {\mathbb R^d}$ based on the 2-layer estimator ${\rm{MMS}}{{\rm{E}}_{k,2,n}}\left( {{\bf{z}}|{\bf{v}}} \right)$, and the Barron constant of the conditional expectation of $\bf{z}$ given $\bf{v}$ \cite[Theorem 1]{diaz2021lower}. We extend this probabilistic bound to explore the MSE performance of DNN-based estimators.

\begin{assumption}
	\emph{Let $k,l,n \in \mathbb{N}$ and $B \in \mathbb{R}^d$ be a bounded set containing $0$. If each entry of $\bf{z}$ belongs to $\left[-1,1\right]$, $\bf{v}$ is supported on $B$, and the conditional expectation $\eta \left( {\bf{v}} \right): = {\mathbb E}\left[ {{\bf{z}}|{\bf{v}}} \right]$ belongs to ${\Gamma _B}$, the set of all function $h:B \to {\mathbb{R}^d}$, then, with probability at least $1-\delta$,}
	\begin{equation}
		{\rm{MMS}}{{\rm{E}}_{k,l,n}}\left( {{\bf{z}}|{\bf{v}}} \right) - {\varepsilon _{k,l,n,\delta }} \le {\rm{MMSE}}\left( {{\bf{z}}|{\bf{v}}} \right)
	\end{equation}
\emph{where}
\begin{equation}
\begin{aligned}
	{\rm{MMSE}}\left( {{\bf{z}}|{\bf{v}}} \right):&=\mathop {\inf }\limits_{h~{\rm{meas}}{\rm{.}}} \mathbb{E}\left[ {{{\left( {{\bf{z}} - h\left( {\bf{v}} \right)} \right)}^2}} \right]\\
	& = \mathbb{E}\left[ {{{\left( {{\bf{z}} - \eta \left( {\bf{v}} \right)} \right)}^2}} \right].
\end{aligned}
\end{equation}
\label{MMSE Bound}
\end{assumption}

\begin{theorem}
	\emph{Under {\bf{Assumption \ref{MMSE Bound}}}, for a linear estimation problem, given a limit $k,n$, and a specific training set $D_n$, if $k \geq d$, and ${\rm{MMS}}{{\rm{E}}_{k,2,n}}\left( {{\bf{z}}|{\bf{v}}} \right) \ne {\rm{MMSE}}\left( {{\bf{z}}|{\bf{v}}} \right)$}, then
		\begin{equation}
		{\varepsilon _{k,2,n,\delta }} \le {\varepsilon _{k,3,n,\delta }} \le  \ldots  \le {\varepsilon _{k,l,n,\delta }}.
	\end{equation}
\label{David_Bound}
\end{theorem}

\begin{proof}
According to {\bf{Assumption \ref{MMSE Bound}}}, for a 2-layer neural network,	with probability at least $1-\delta$,
\begin{equation}
	{\rm{MMS}}{{\rm{E}}_{k,2,n}}\left( {{\bf{z}}|{\bf{v}}} \right) - {\varepsilon _{k,2,n,\delta }} < {\rm{MMSE}}\left( {{\bf{z}}|{\bf{v}}} \right).
\end{equation}
If ${\rm{MMS}}{{\rm{E}}_{k,2,n}}\left( {{\bf{z}}|{\bf{v}}} \right) \ne {\rm{MMSE}}\left( {{\bf{z}}|{\bf{v}}} \right)$, the output of the trained 2-layer neural network $h_{k,2}^ * \left( {\bf{v}} \right)$ is not a \emph{sufficient statistic} for estimation of $\bf{z}$ from $\bf{v}$. By the data processing inequality, we have 
\begin{equation}
	I\left({\bf{z}};h_{k,2}^ * \left( {\bf{v}} \right)\right) \geq I\left({\bf{z}};h_{k,3}^ * \left( {\bf{v}} \right)\right) \geq ... \geq I\left({\bf{z}};h_{k,l}^ * \left( {\bf{v}} \right)\right),
\end{equation}
and
\begin{equation}
	{\varepsilon _{k,2,n,\delta }} \le {\varepsilon _{k,3,n,\delta }} \le  \ldots  \le {\varepsilon _{k,l,n,\delta }}.
\end{equation}
\end{proof}

\begin{figure*}[ht]
	\centering
	\subfigure[SNR=0~dB]{              
		\includegraphics[width=3.50in]{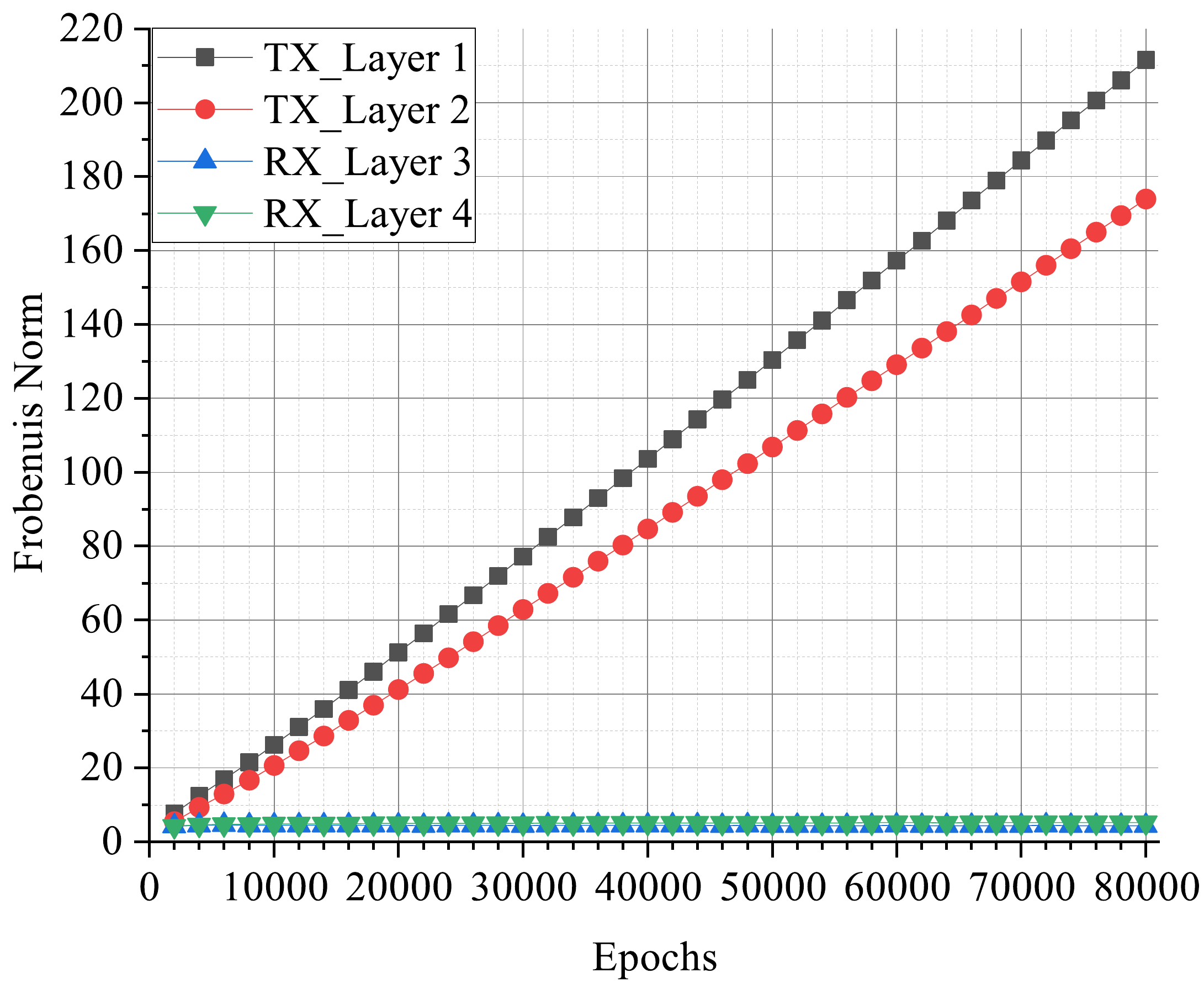}}
	\subfigure[SNR=25~dB]{
		\includegraphics[width=3.50in]{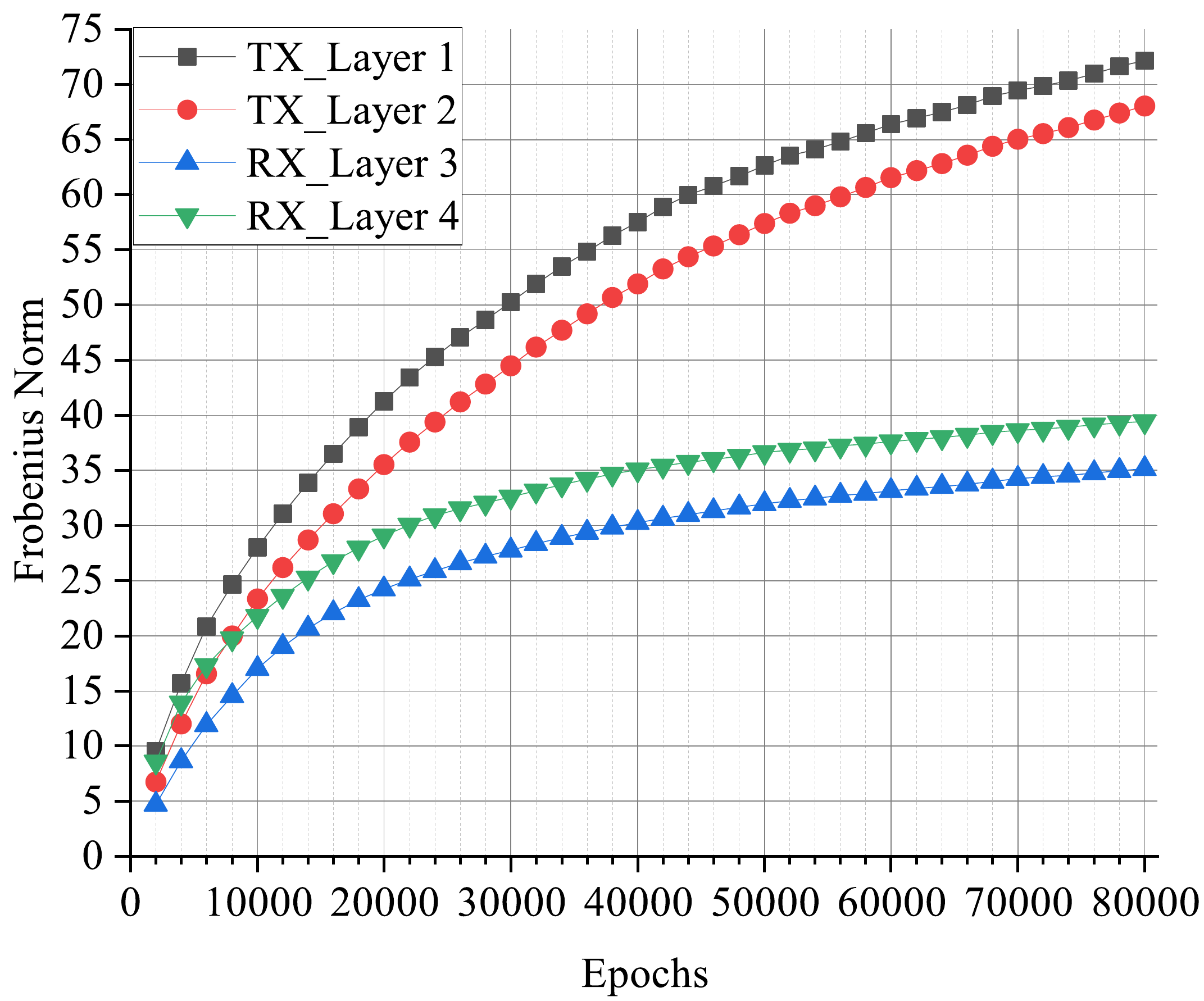}}
	\caption{The Frobenius norm of each layer of a AE-based communication system ${\left\| {{{\bf{W}}^{\left( h \right)}}\left( k \right)} \right\|_F}$ versus epochs under Gaussian channel.}
	\label{FNorm_Gaussian}
\end{figure*}

\subsection{Information Flow in Neural Networks}
If the joint probability distribution $p\left( {{\bf{v}},{\bf{z}}} \right)$ is known, the expected (population) risk ${{\cal C}_{p\left( {{\bf{v}},{\bf{z}}} \right)}}\left( {{g_{{{\bf{\Theta }}_g}}},{{\cal L}_{\log }}} \right)$ can be written as
\begin{equation}
	\begin{aligned}
		{\mathbb{E}}\left[ {{{\cal L}_{\log }}\left( {{\bf{\hat z}},{\bf{z}};{{\bf{\Theta }}_g}} \right)} \right] &= \sum\limits_{{\bf{v}} \in \mathcal V,{\bf{z}} \in \mathcal Z} {p\left( {{\bf{v}},{\bf{z}}} \right)\log \left( {\frac{1}{{Q\left( {{\bf{z}}|{\bf{v}}} \right)}}} \right)}\\
		&=\sum\limits_{{\bf{v}} \in \mathcal V,{\bf{z}} \in \mathcal Z} {p\left( {{\bf{v}},{\bf{z}}} \right)\log \left( {\frac{1}{{p\left( {{\bf{z}}|{\bf{v}}} \right)}}} \right)} {\rm{ + }}\\
		&~~~\sum\limits_{{\bf{v}} \in \mathcal V,{\bf{z}} \in \mathcal Z} {p\left( {{\bf{v}},{\bf{z}}} \right)\log \left( {\frac{{p\left( {{\bf{z}}|{\bf{v}}} \right)}}{{Q\left( {{\bf{z}}|{\bf{v}}} \right)}}} \right)}\\
		&=H\left( {{\bf{z}}|{\bf{v}}} \right) + {D_{{\rm{KL}}}}\left( {p\left( {{\bf{z}}|{\bf{v}}} \right)||Q\left( {{\bf{z}}|{\bf{v}}} \right)} \right)\\
		& \ge H\left( {{\bf{z}}|{\bf{v}}} \right) 
	\end{aligned}
\label{PopulationRisk}
\end{equation}
where $Q\left( { \cdot |{\bf{v}}} \right){\rm{ = }}{g_{{{\bf{\Theta }}_g}}}\left( {\bf{v}} \right) \in p\left( {\mathcal{Z}} \right)$ and ${D_{{\rm{KL}}}}\left( {p\left( {{\bf{z}}|{\bf{v}}} \right)||Q\left( {{\bf{z}}|{\bf{v}}} \right)} \right)$ denotes Kullback-Leibler divergence between ${p\left( {{\bf{z}}|{\bf{v}}} \right)}$ and ${Q\left( {{\bf{z}}|{\bf{v}}} \right)}$ \cite{zaidi2020information}\footnote{If $\bf{z}$ and $\bf{v}$ are continuous random variables, the sum becomes an integral when their PDFs exist.}. If and only if the decoder is given by the conditional posterior ${g_{{{\bf{\Theta }}_g}}}\left( {\bf{v}} \right){\rm{ = }}p\left( {{\bf{z}}|{\bf{v}}} \right)$, the expected (population) risk reaches the minimum $\mathop {\min }\limits_{{g_{{{\bf{\Theta }}_g}}}} {{\cal C}_{p\left( {{\bf{v}},{\bf{z}}} \right)}}\left( {{g_{{{\bf{\Theta }}_g}}},{{\cal L}_{\log }}} \right) = H\left( {{\bf{z}}|{\bf{v}}} \right)$.

In physical layer communication, instead of perfectly knowing the channel-related joint probability distribution $p\left( {{\bf{v}},{\bf{z}}} \right)$, we only have a set of $n$ \textit{i.i.d.} samples ${D_n}: = \left\{ {\left( {{{\bf{v}}_i},{{\bf{z}}_i}} \right)} \right\}_{i = 1}^n$ from $p\left( {{\bf{v}},{\bf{z}}} \right)$. In this case, the empirical risk is defined as
\begin{equation}
	{\hat {\cal C}_{p\left( {{\bf{v}},{\bf{z}}} \right)}}\left( {{g_{{{\bf{\Theta }}_g}}},{\cal L},{{\cal D}_n}} \right) = \frac{1}{n}\sum\limits_{i = 1}^n {{\cal L}\left[ {{{\bf{z}}_i},{g_{{{\bf{\Theta }}_g}}}\left( {{{\bf{v}}_i}} \right)} \right]}.
\end{equation}
Practically, the ${\mathcal D}_n$ from $p\left( {{\bf{v}},{\bf{z}}} \right)$ usually is a finite set. This leads the difference between the empirical and expected (population) risks which can be defined as
\begin{equation}
	\begin{aligned}
		{\rm{ge}}{{\rm{n}}_{p\left( {{\bf{v}},{\bf{z}}} \right)}}\left( {{g_{{{\bf{\Theta }}_g}}},{\cal L},{{\cal D}_n}} \right)=&{{\cal C}_{p\left( {{\bf{v}},{\bf{z}}} \right)}}\left( {{g_{{{\bf{\Theta }}_g}}},{{\cal L}_{\log }}} \right)-\\
		&{\hat {\cal C}_{p\left( {{\bf{v}},{\bf{z}}} \right)}}\left( {{g_{{{\bf{\Theta }}_g}}},{\cal L},{{\cal D}_n}} \right).
	\end{aligned}
\end{equation}

We now can preliminarily conclude that the DNN-based receiver is an estimator with minimum empirical risk for a given set ${\mathcal D}_n$, whereas its performance is inferior to the optimal with minimum expected (population) risk under a given joint probability distribution $p\left( {{\bf{v}},{\bf{z}}} \right)$.

Furthermore, it is necessary to quantitatively understand how information flows inside the decoder. Fig. \ref{Decoder}(b) shows the graph representation of the decoder where ${{{\bf{t}}_i}}$ and ${{{\bf{t'}}}_i}\left( {1 \le i \le S} \right)$ denote $i$-th hidden layer representations starting from the input layer and the output layer, respectively. Usually, the Shannon's entropy cannot be calculated directly since the exact joint probability distribution of two variables is difficult to acquire. Therefore, we use the method proposed in \cite{yu2019understanding} to illustrate layer-wise mutual information by three kinds of information planes (IPs) where the Shannon’s entropy is estimated by matrix-based functional of Rényi’s $\alpha$-entropy \cite{giraldo2014measures}. Its details are given in Appendix.

\section{Simulation Results}
\label{Simulation Results}
In this section, we provide simulation results to illustrate the behaviour of DNN in physical layer communication.
\subsection{Constellation and Convergence of AE-based Communication System}
\subsubsection{Gaussian Channel}
\begin{figure}[t]
	\centering
	\includegraphics[height=0.40\textwidth]{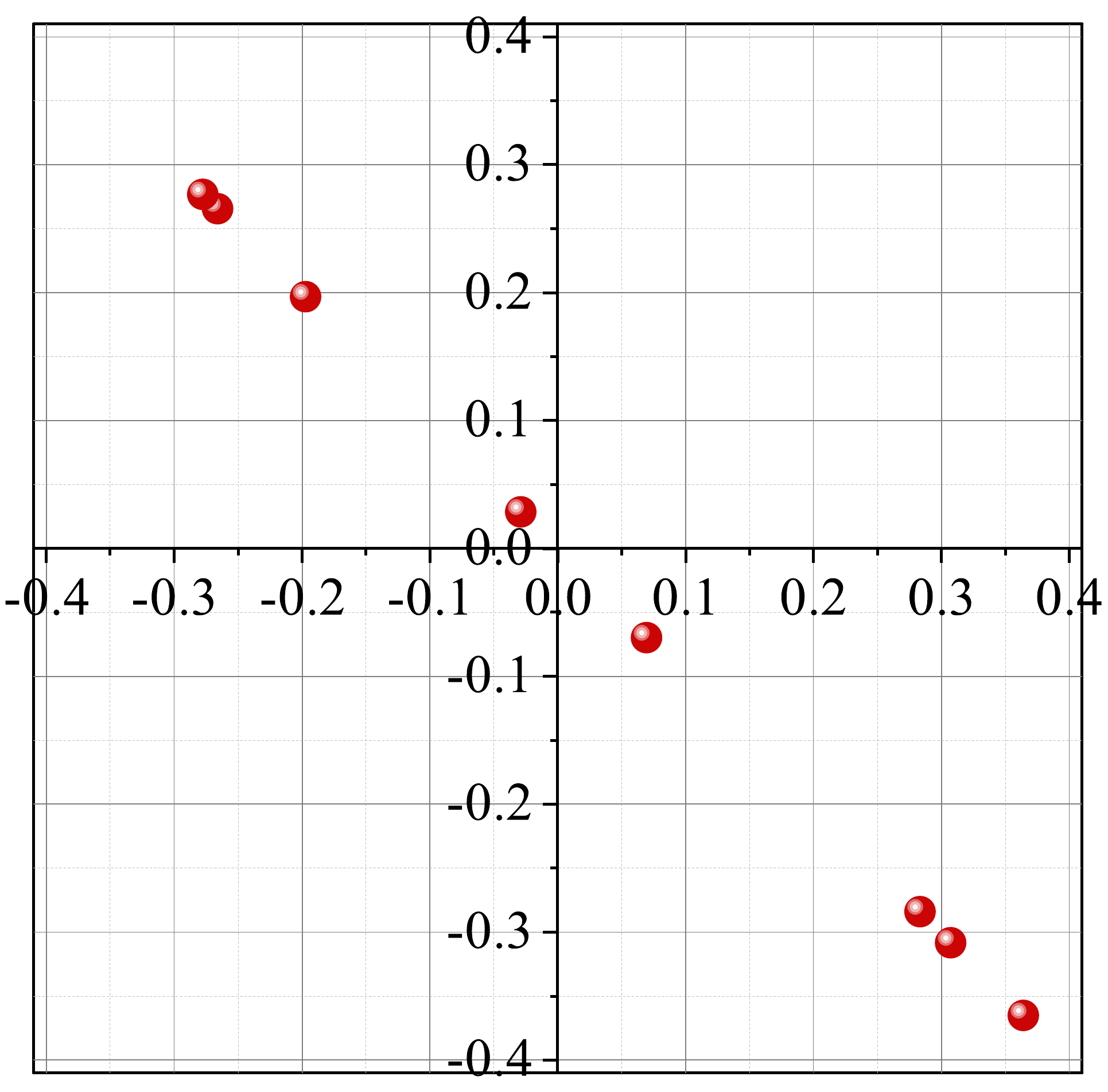}
	\caption{Constellation produced by AE for $d=2$ and $M=8$ under Gaussian channel (SNR = 0 dB).}
	\label{AE2D8_0dB}
\end{figure}

\begin{table}[t]
	\centering
	\caption{Layout of the autoencoder}
	\begin{tabular}{l|c}
		Layer         & Output dimensions \\ \hline
		Input         & $M$                 \\
		Dense+ReLU    & $M$                 \\
		Dense+linear  & $d$                 \\
		Normalization & $d$                 \\ \hline
		Channel         & $d$                 \\ \hline
		Dense+ReLU    & $M$                 \\
		Dense+softmax & $M$                
	\end{tabular}
	\label{AEsLayout}
\end{table}

Fig. \ref{FNorm_Gaussian}(a) and Fig. \ref{FNorm_Gaussian}(b) visualize the Frobenius norm of each layer of a AE-based communication system for $d=2$ and $M=8$ versus epochs under Gaussian channel. The Layout of the AEs are provided in Table \ref{AEsLayout}. When SNR=0~dB, the Frobenius norms of layers at the transmitter side increase with the epoch, and that at the receiver side are kept in small values and do not change significantly. In contrast to the case of low SNR, the Frobenius norms of all the layers close to convergence after $4 \times10^4$ epochs when SNR=25~dB. This phenomenon can be explained by our analysis in Section \ref{Section Encoder}. If SNR is low, large biases would be introduced into the channel layer by noise since ${{{\bf{W'}}}^{\left( {{H_1} + 1} \right)}} = \left[ {{{\bf{I}}_m},{\bf{n}}} \right]$. This leads to the \emph{exploding gradient problem} at the transmitter side. At the receiver side, the Frobenuis norm of each layer tends to be small to counter the large noise, but it produces a very little effect. From Fig. \ref{AE2D8_0dB}, it can be seen that all the signal points are on a line, and two of them are almost overlapped. A number of simulations have been conducted, and the constellations generated by the AE with different initial parameters keep in the same pattern that all constellation points are in a straight line.
\begin{figure}[t]
	\centering
	\subfigure[]{
		\includegraphics[width=0.48\textwidth]{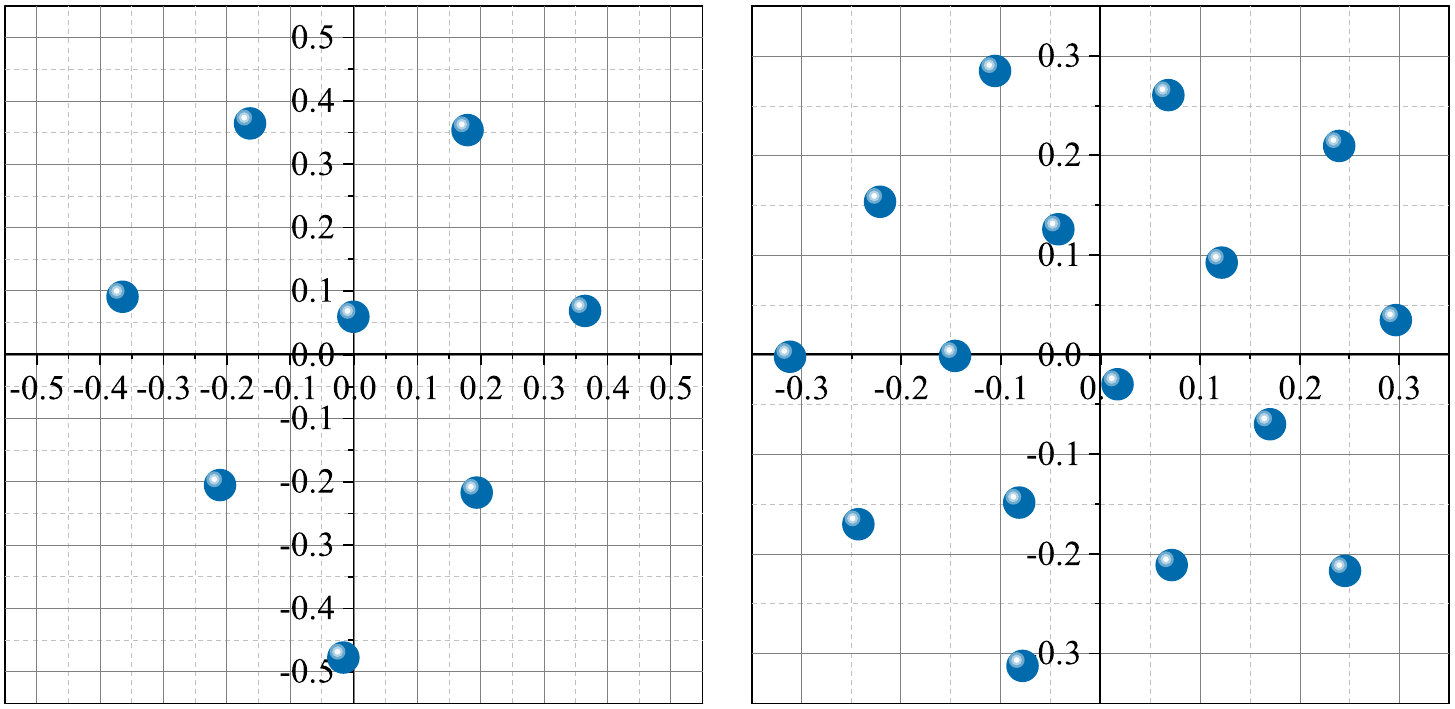}
	}
	\quad
	\subfigure[]{
		\includegraphics[width=0.48\textwidth]{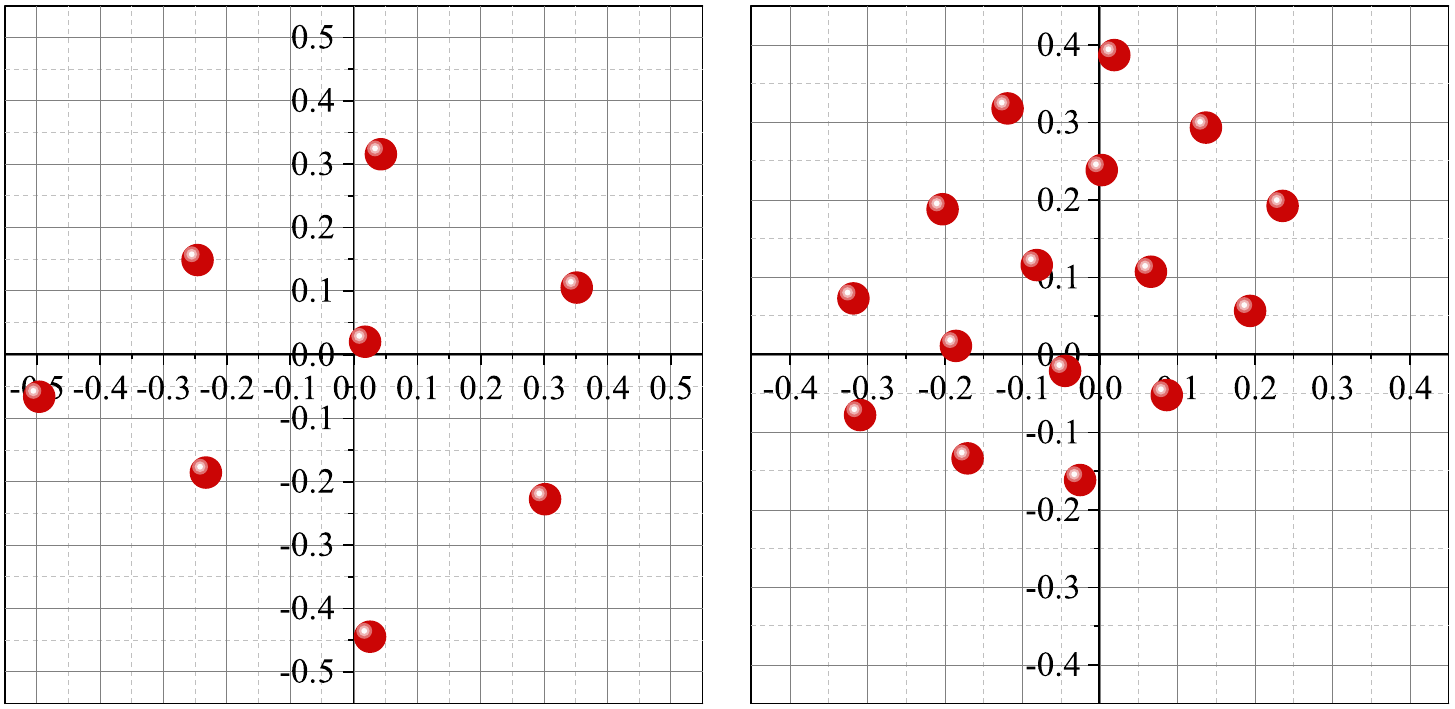}
	}
	\caption{Comparisons of (a) optimum constellation obtained by gradient-search technique and (b) constellation produced by AEs for $d=2$ and $M=8~\rm{or}~16$.}
	\label{Constellations2D}
\end{figure}

\begin{figure}[t]
	\centering
	\subfigure[]{
		\includegraphics[width=0.20\textwidth]{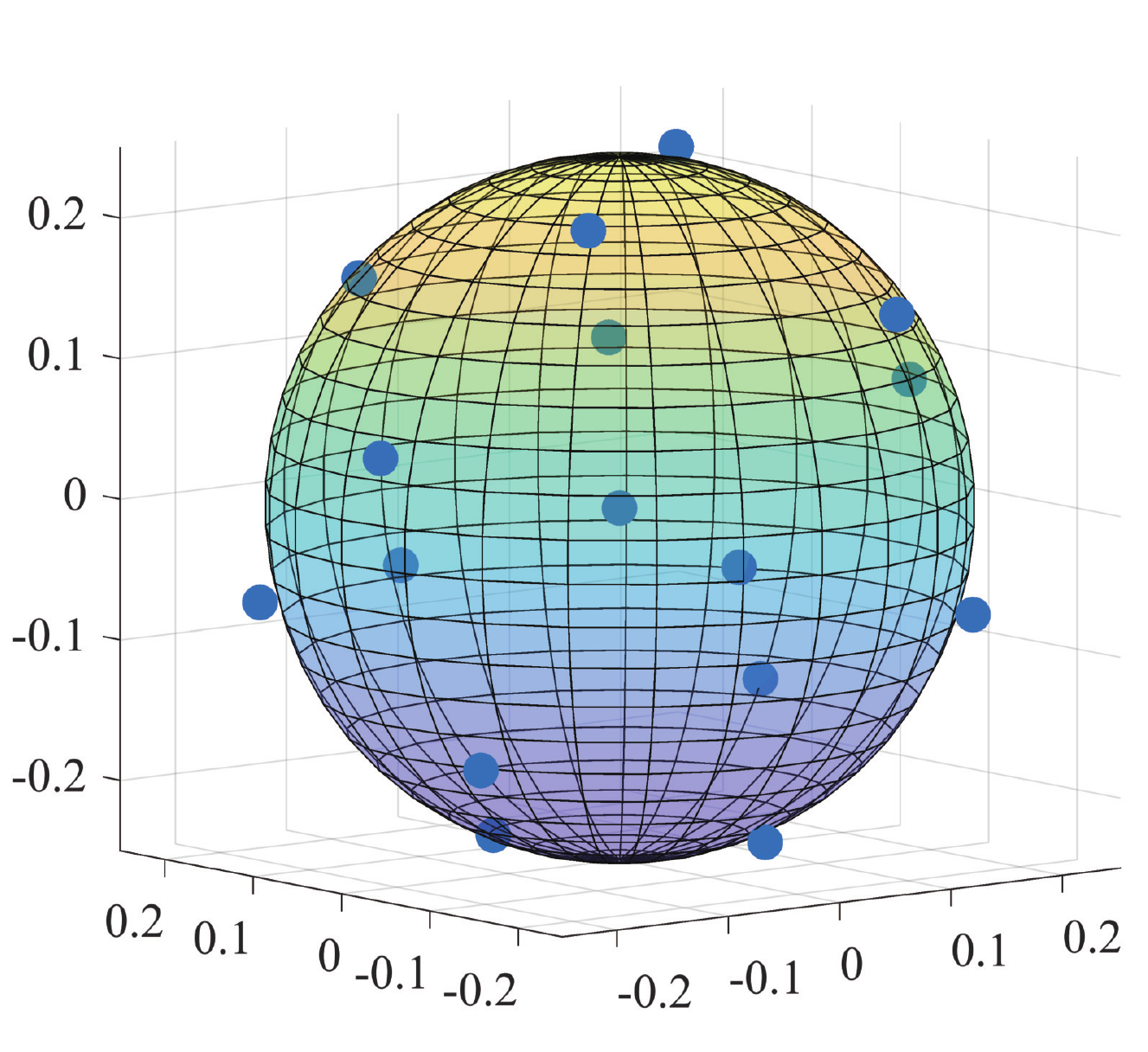}
	}\subfigure[]{
		\includegraphics[width=0.25\textwidth]{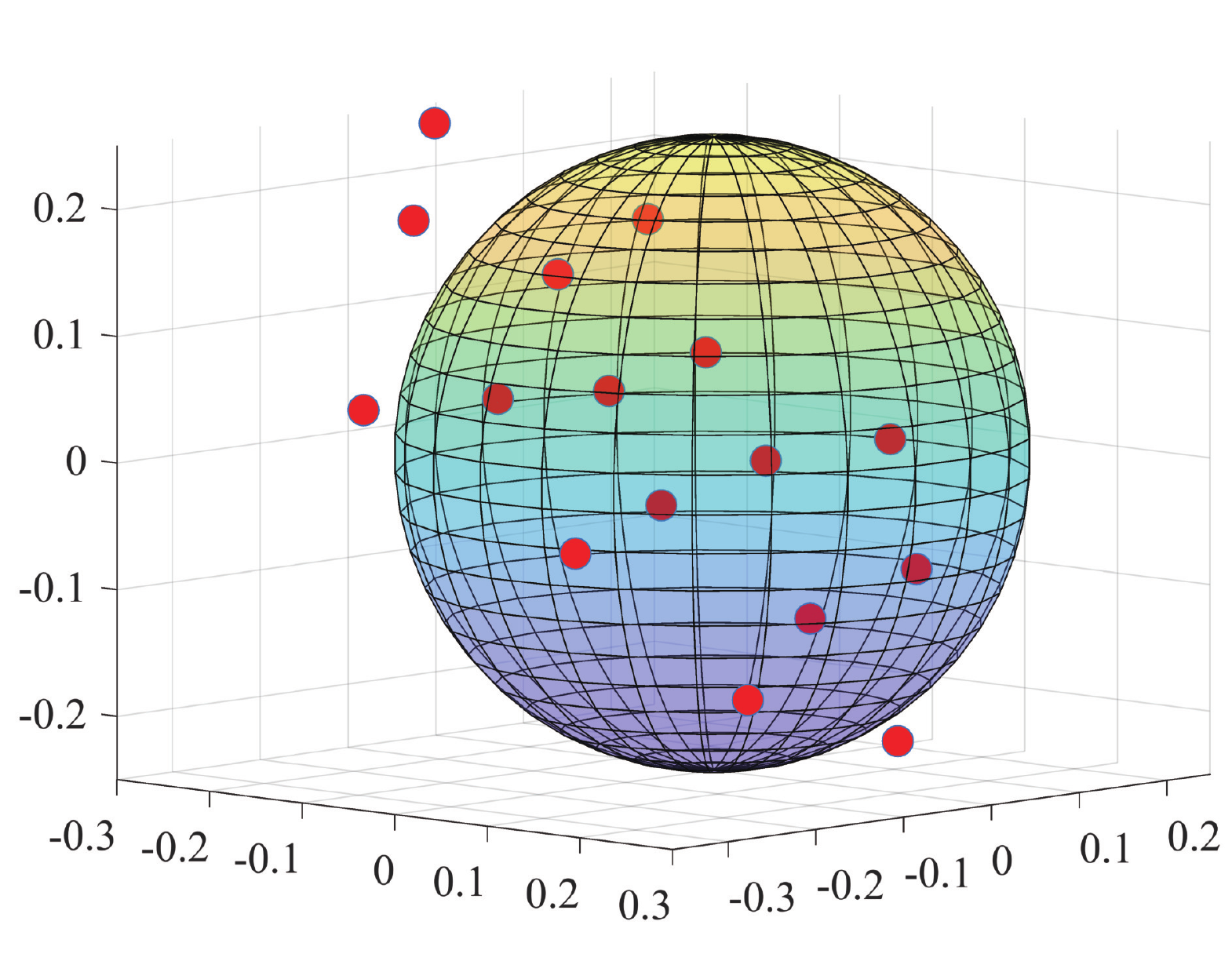}
	}
	\caption{Comparisons of (a) optimum constellation obtained by gradient-search technique and (b) constellation produced by AEs for $d=3$ and $M=16$.}
	\label{Constellations3D}
\end{figure}

\begin{figure}[t]
	\centering
	\subfigure[Frobenius norm of each layer ${\left\| {{{\bf{W}}^{\left( h \right)}}\left( k \right)} \right\|_F}$]{
		\includegraphics[height=0.40\textwidth]{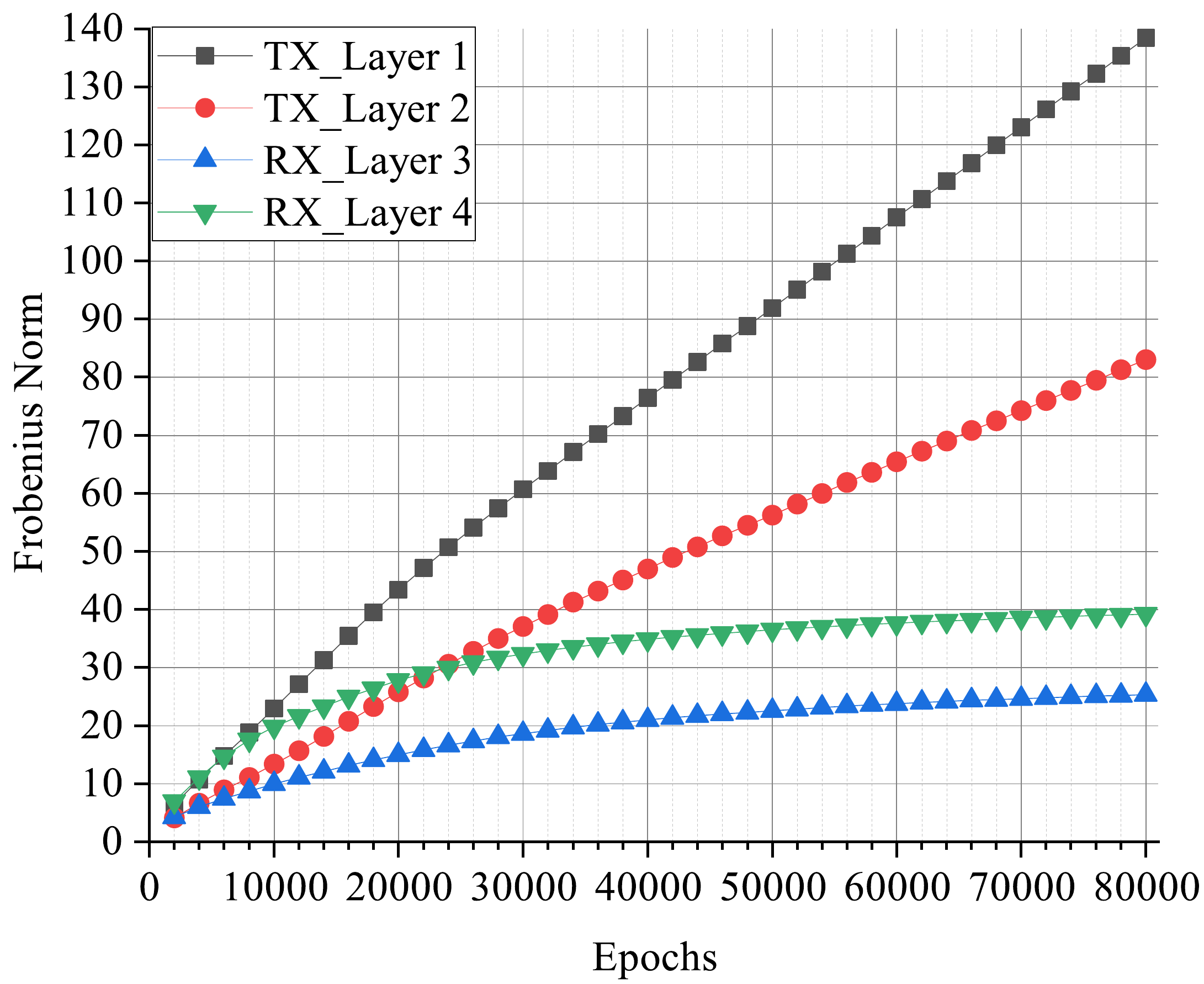}
	}
	\quad
	\subfigure[Constellation]{
		\includegraphics[height=0.40\textwidth]{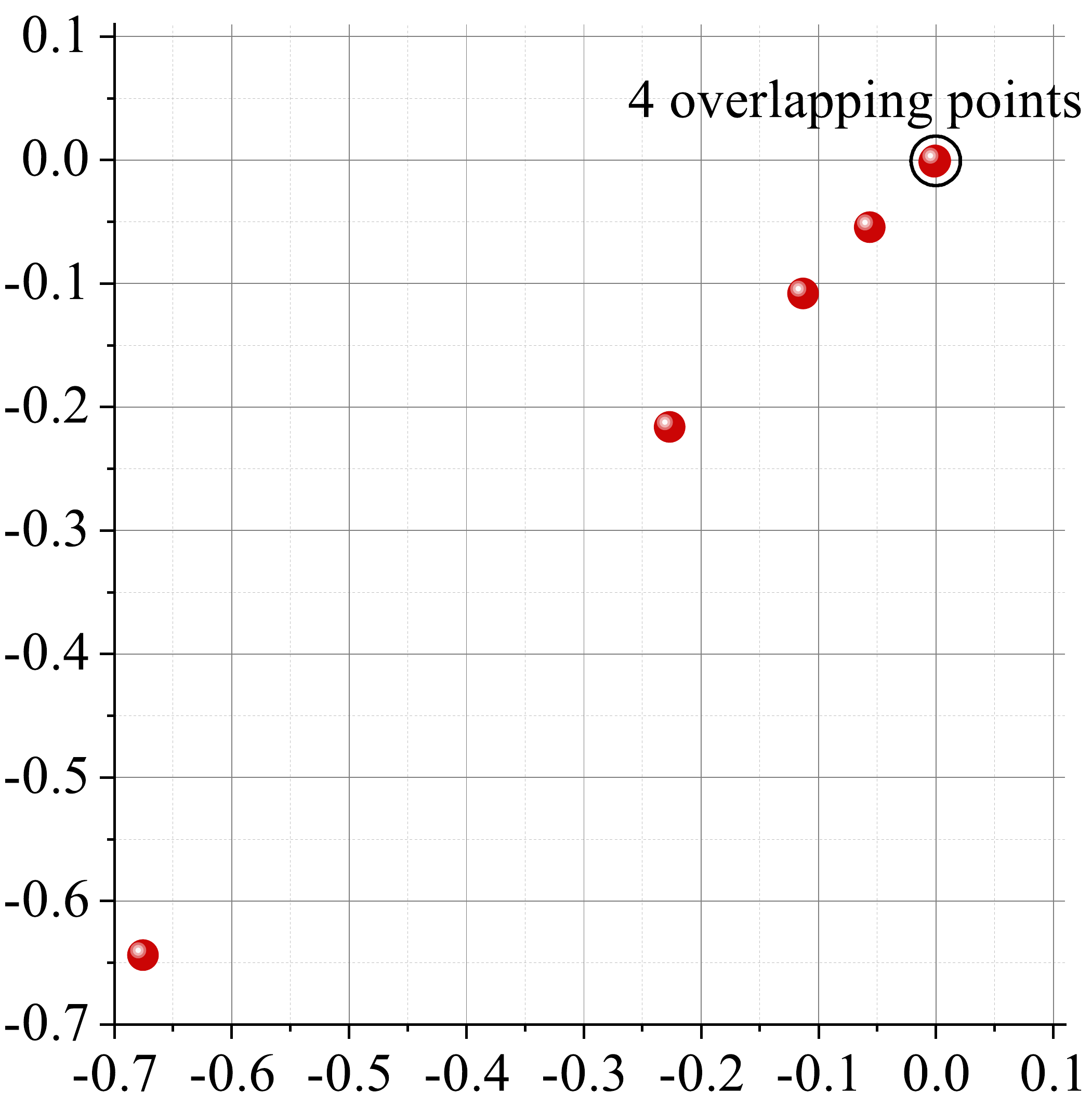}
	}
	\caption{The simulation results of the AE for $d=2$ and $M=8$ under Rayleigh flat fading channel (SNR = 25dB).}
	\label{FNorm&Con_RayleighFadingChannel}
\end{figure}

Fig. \ref{Constellations2D}(a) and Fig. \ref{Constellations3D}(a) show the optimum constellations obtained by gradient-search technique proposed in \cite{foschini1974optimization}. When $d=2$ and $3$, the algorithm was run allowing for 1000 and 3000 steps, respectively. Several random initialization constellations (initial points selected according to a uniform probability density of the unit disk) are used for each value of $d$ and $M$. The step size $\eta  = 2 \times {10^{ - 4}}$. Numerous local optima are found, many of which are merely rotations or other symmetric modifications of each other. Fig. \ref{Constellations2D}(b) and Fig. \ref{Constellations3D}(b) show the constellations produced by AEs which have the same network structures and hyperparameters with the AEs mentioned in \cite{o2017introduction} (see also TABLE \ref{AEsLayout}). The AEs were trained with ${10^6}$ epochs, each of which contains $M$ different symbols. Several simulations have been conducted, and it can be found that AE does not ensure that the optimum constellation whereas the gradient-search technique has a high probability of finding an optimum.

When $d=2$, the two-dimensional constellations produced by AEs have a similar pattern to the optimum constellations which form a lattice of (almost) hexagonal. Specifically, in the case of $\left( {d = 2,\;M = 8} \right)$, one of the constellations found by the AE can be obtained by rotating the optimum constellation found by gradient-search technique. In the case of $\left( {d = 2,\;M = 16} \right)$, the constellation found by the AE is different from the optimum constellation but it still forms a lattice of (almost)  equilateral triangles. In the case of $\left( {d = 3,~M = 16} \right)$, one signal point of the optimum constellation found by gradient-search technique is almost at the origin while the other 15 signal points are almost at the surface of a sphere with radius $P_{{\rm{av}}}$ and centre 0. This pattern is similar to the surface of a truncated icosahedron which is composed of pentagonal and hexagonal faces. However, the three-dimensional constellation produced by an AE is a local optima which is formed by 16 signal points almost in a plane.  

From the perspective of computational complexity, the cost to train an AE is significantly higher than the cost of traditional algorithm. Specifically, an AE which has 4 dense layers respectively with $M$, $d$, $M$ and $M$ neurons needs to train $\left( {2M + 1} \right)\left( {M + d} \right) + 2M$ parameters for ${10^6}$ epochs whereas the gradient-search algorithm only needs $2M$ trainable parameters for ${10^3}$ steps.

\subsubsection{Rayleigh Flat Fading Channel}
In the case of Rayleigh flat fading channel, multiplicative perturbation is introduced since ${{{\bf{W'}}}^{\left( {{H_1} + 1} \right)}} = \left[ {{\bf{H}},{\bf{n}}} \right]$. Fig. \ref{FNorm&Con_RayleighFadingChannel}(a) illustrates that the \emph{exploding gradient problem} occurs at the side of transmitter under the case of Rayleigh flat fading channel with SNR= 25dB, which is similar to the case of low SNR Gaussian channel. It means that fading leads the AE not to convergence even if noise is small. Fig. \ref{FNorm&Con_RayleighFadingChannel}(b) illustrates the corresponding constellation produced by AE. The receiver cannot distinguish all the transmitted signals  correctly because of 4 overlapping points.

Summarily, the structure of AE-based communication system makes it has strict requirements on channel to train the network. First, it demands high SNR. Second, the AE cannot work properly in fading channels. These impede the implement of AE-based communication system in practical wireless scenarios.

\begin{figure}[t]
	\centering
	\includegraphics[width=3.50in]{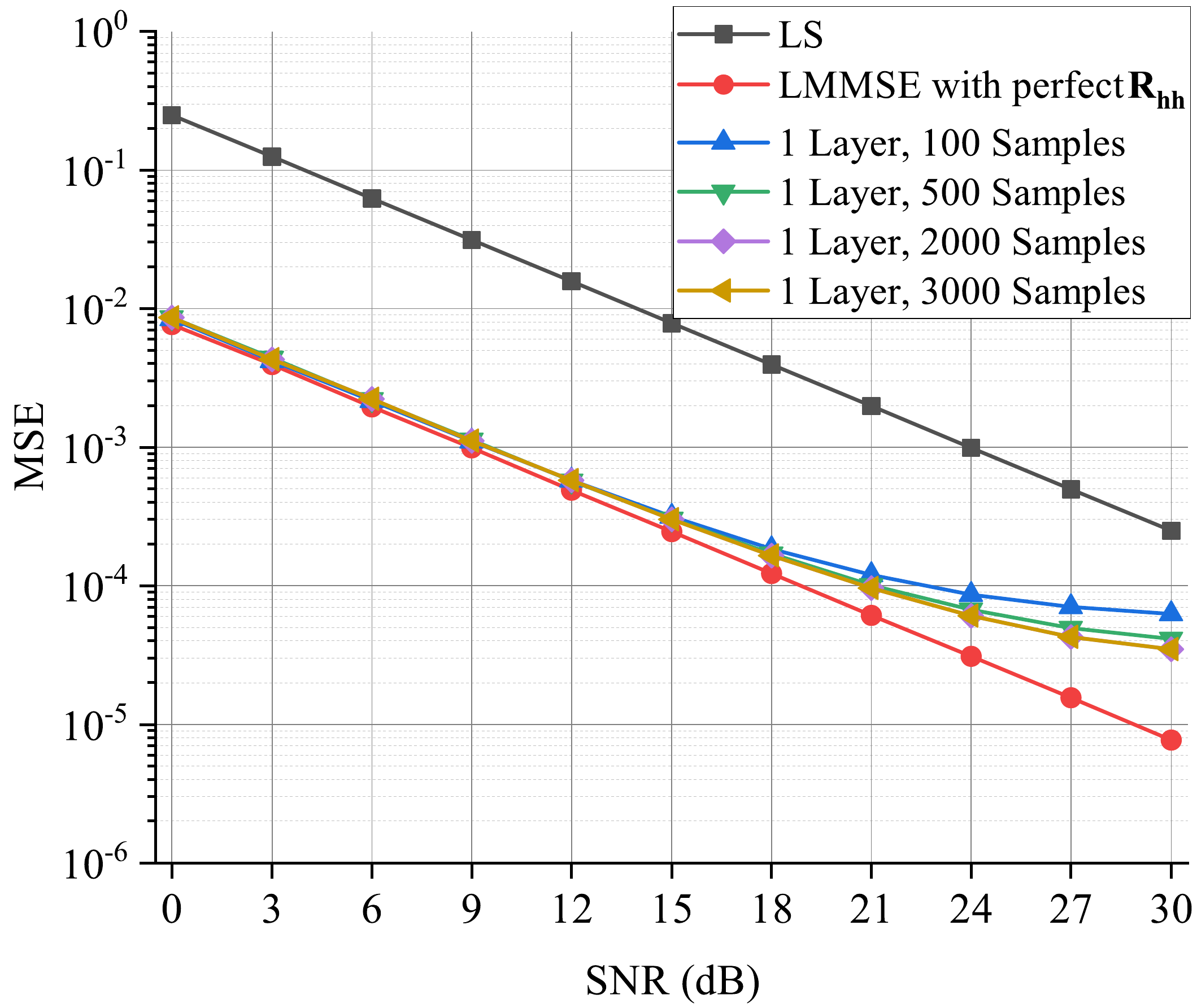}
	\caption{The MSE performance of LS estimator, LMMSE estimator, and single hidden layer estimator trained by different training sets versus SNRs.}
	\label{MSE_Samples}
\end{figure}

\subsection{The Performance of DNN-based Estimation}
We consider a common channel estimation problem for an OFDM system. Let ${\bf{z}} \buildrel \Delta \over = {\left[ {H\left[ 0 \right],H\left[ 1 \right], \cdots ,H\left[ {N_c - 1} \right]} \right]^T}$ which denotes channel frequency response (CFR) vector of a channel. $N_c$ denotes the number of subcarriers of an OFDM symbol. For the sake of convenience, we denote the measurable variable as ${\bf{v}} \buildrel \Delta \over = {{\bf{\hat z}}_{{\rm{LS}}}}$ where ${{\bf{\hat z}}_{{\rm{LS}}}}$ represents the LS estimation of $\bf{z}$. Usually, it can be obtained by using training symbol-based channel estimation. Practically, the MMSE channel estimation will not be chosen, unless the covariance matrix of fading channel is known.

\subsubsection{MSE Performance}
Fig. \ref{MSE_Samples} compares the MSEs of the LS, LMMSE, and single
 hidden layer estimators versus SNRs. The size of the test set is $10^5$. The size of neural network $d=128$. The MSEs of the LS and LMMSE estimators are used as the benchmarks. The MSE performance of the single hidden layer neural network improves as the size of the training set increases. Compared to LS estimator, the single hidden layer neural network has superior performance, even if the size of the training set is 100. However, its performance is inferior to the LMMSE regardless of the size of the training set.

\begin{figure}[t]
	\centering
	\includegraphics[width=3.50in]{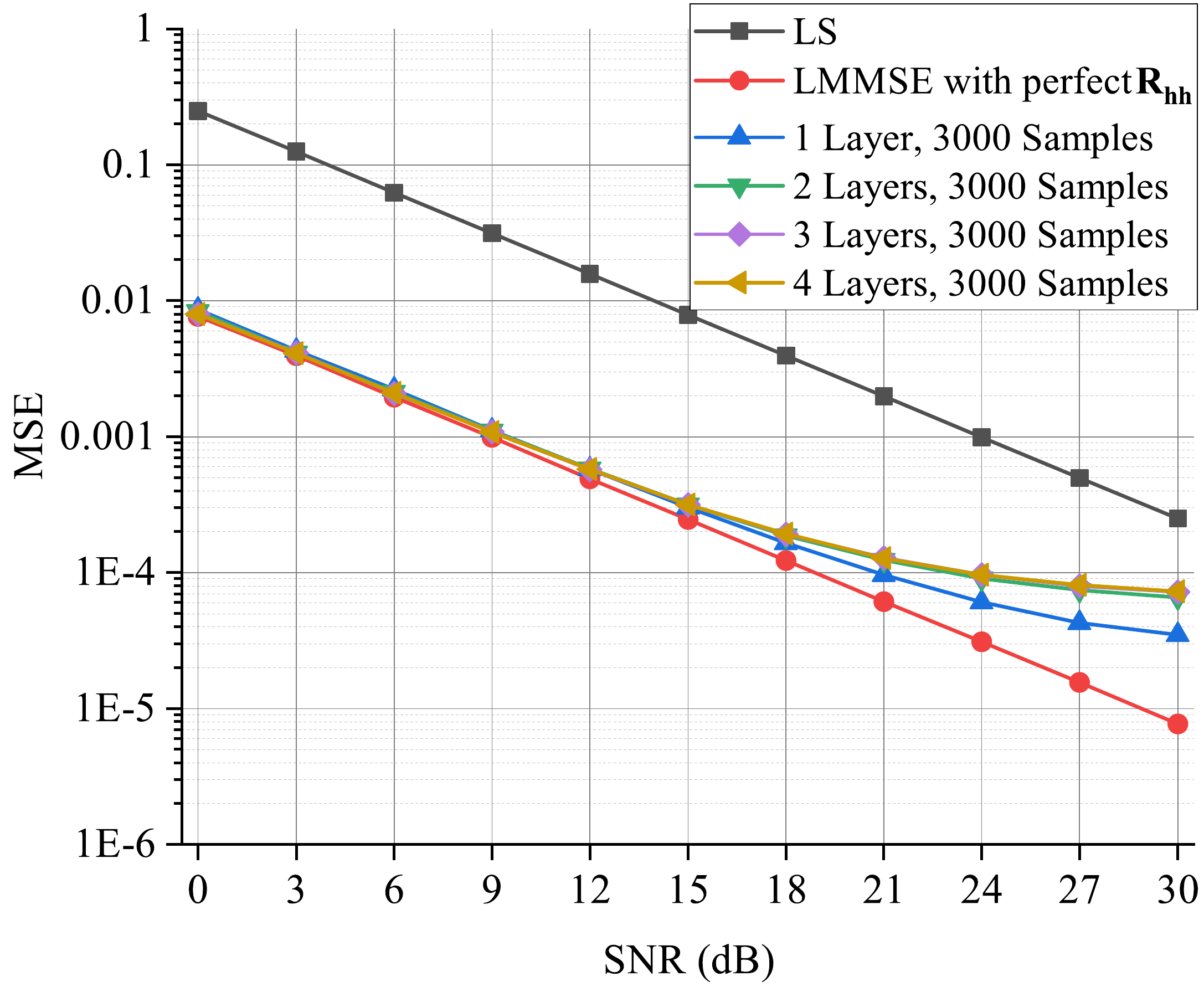}
	\caption{The MSE performance of LS estimator, LMMSE estimator, and neural network-based estimators with different hidden layers versus SNRs.}
	\label{MSE_Layers}
\end{figure}

\begin{figure}[t]
	\centering
	\includegraphics[width=3.50in]{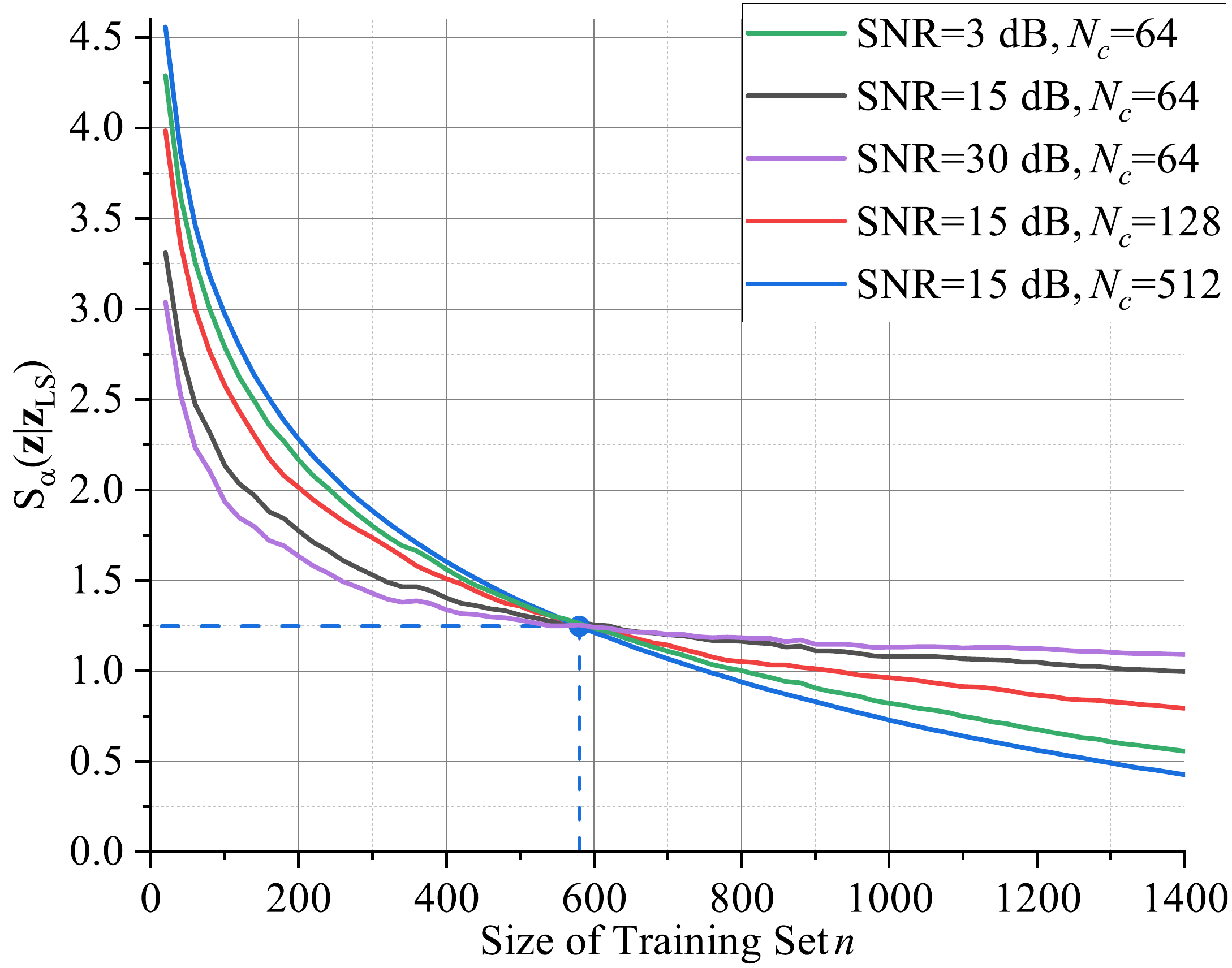}
	\caption{The entropy ${S_\alpha }\left( {{\bf{z}}|{{{\bf{\hat z}}}_{{\rm{LS}}}}} \right)$  with respect to different values of SNR and $N_c$.}
	\label{PopulationRiskFig}
\end{figure}

Fig. \ref{MSE_Layers} compares the MSEs of neural network-based estimators with different hidden layers to that of the LS and LMMSE estimators. The sizes of these neural networks are the same with $d=128$. Given a limited set of training data, the MSE performance of neural network-based estimator degrades with increased hidden layers. This result coincides with the description in {\bf{Theorem \ref{David_Bound}}}.

\subsubsection{Information Flow}

According to (\ref{PopulationRisk}), the minimum logarithmic expected (population) risk for this inference problem is $H\left( {{\bf{z}}|{{{\bf{\hat z}}}_{{\rm{LS}}}}} \right)$ which can be estimated by Rényi's $\alpha$-entropy ${S_\alpha }\left( {{\bf{z}}|{{{\bf{\hat z}}}_{{\rm{LS}}}}} \right){\rm{ = }}{S_\alpha }\left( {{\bf{z}},{{{\bf{\hat z}}}_{{\rm{LS}}}}} \right) - {S_\alpha }\left( {{{{\bf{\hat z}}}_{{\rm{LS}}}}} \right)$ with $\alpha=1.01$. Fig. \ref{PopulationRiskFig} illustrates the entropy ${S_\alpha }\left( {{\bf{z}}|{{{\bf{\hat z}}}_{{\rm{LS}}}}} \right)$ with respect to different values of SNR and $N_c$. In a practical scenario, we use linear interpolation and the number of pilots ${N_p} = N_c/4$. As can be seen, ${S_\alpha }\left( {{\bf{z}}|{{{\bf{\hat z}}}_{{\rm{LS}}}}} \right)$ monotonically decreases as the size of training set increases. When $n \to \infty$, ${S_\alpha }\left( {{\bf{z}}|{{{\bf{\hat z}}}_{{\rm{LS}}}}} \right)$ decreases slowly. It is because the joint probability distribution $p\left( {{\bf{z}},{{{\bf{\hat z}}}_{{\rm{LS}}}}} \right)$ can be perfectly learned and therefore the empirical risk is approaching to the expected risk. Interestingly, when $n>580$, the lower the SNR or the larger input dimension $d$ is, the smaller $n$ is needed to obtain the same value of ${S_\alpha }\left( {{\bf{z}}|{{{\bf{\hat z}}}_{{\rm{LS}}}}} \right)$.

\begin{figure*}[t]
	\centering
	\includegraphics[width=6.0in]{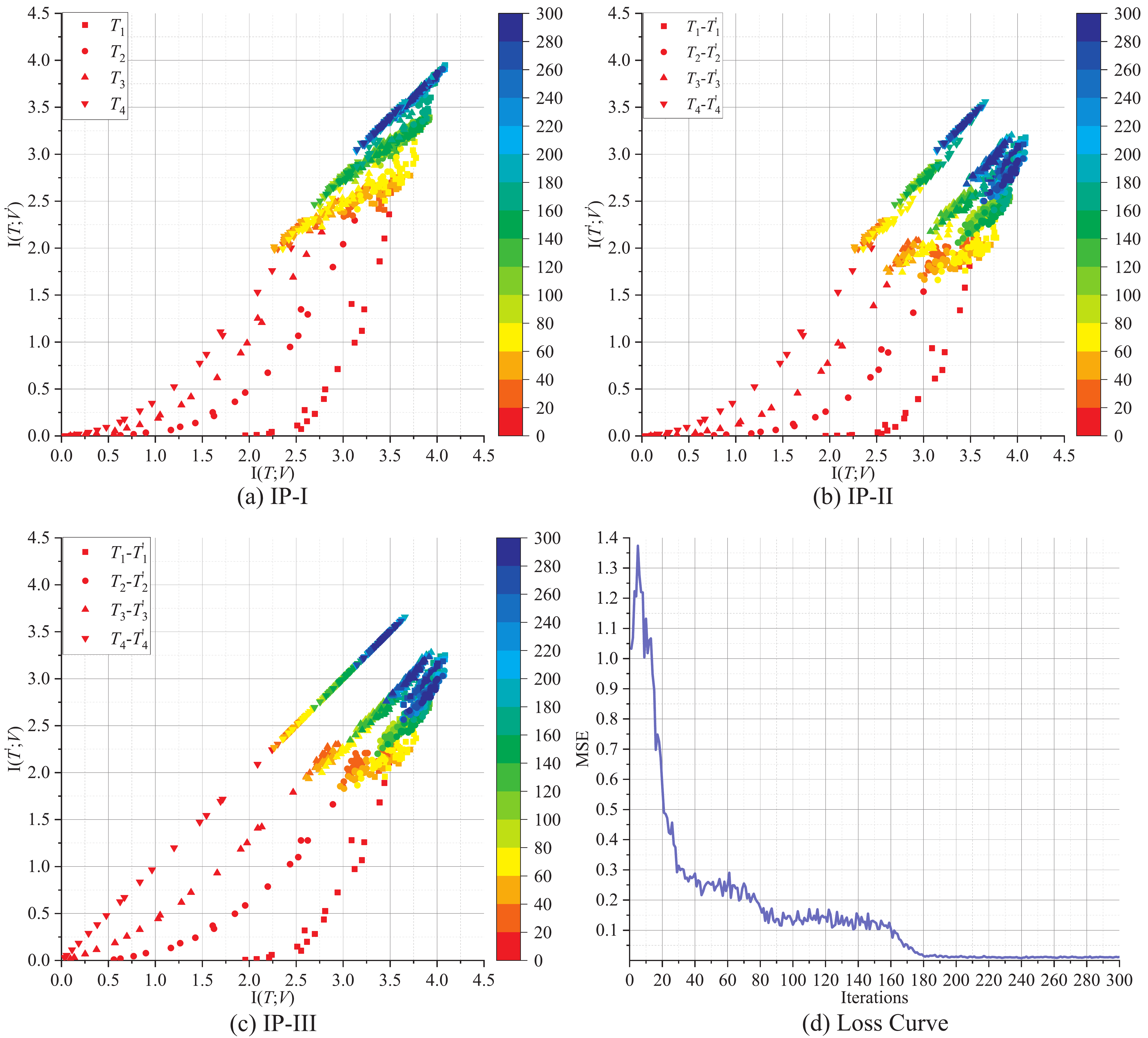}
	\caption{The three IPs and loss curve in a DNN-based OFDM channel estimator for $N_c=64$ and $S=4$.}
	\label{InformationPlane_S=4}
\end{figure*}

\begin{figure*}[t]
	\begin{center}
		\centering
		\includegraphics[width=6.0in]{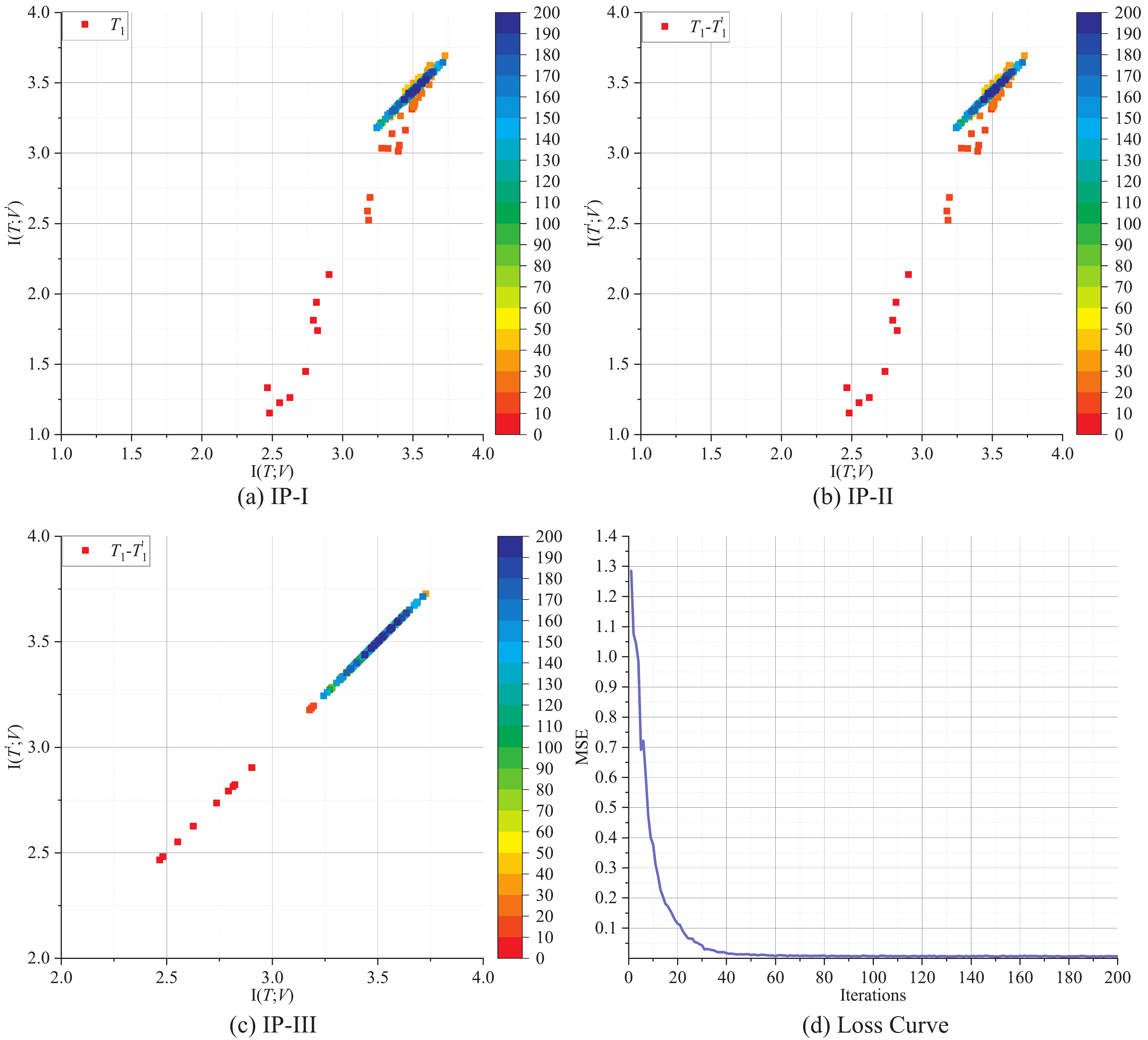}
		\caption{The three IPs and loss curve in a linear SLFN-based OFDM channel estimator for $N_c=64$ and $S=1$.}
		\label{InformationPlane_S=1}
	\end{center}
\end{figure*}

\begin{table}[t]
	\centering
	\caption{Layout of the NN-based OFDM Channel Estimators}
	\begin{tabular}{c|c|c|}
		\multicolumn{1}{l|}{$S$} & \multicolumn{1}{l|}{Input/Output dimension ($2N_c$)} & \multicolumn{1}{l|}{Number of hidden layers ($2S-1$)} \\ \hline
		1                      & 128                                                  & 1                                                   \\
		4                      & 128                                                  & 7                                                  
	\end{tabular}
\label{EstimatorsLayout}
\end{table}

We analyze two NN-based OFDM channel estimators with different layouts (See TABLE \ref{EstimatorsLayout}). Fig. \ref{InformationPlane_S=4}(a), (b) and (c) illustrate the behavior of the IP-I, IP-II and IP-III  in a DNN-based OFDM channel estimator with topology ``$128 - 64 - 32 - 16 - 8 - 16 - 32 - 64 - 128$'' where linear activation function is considered and the training sample is constructed by concatenating the real and imaginary parts of the complex channel vectors. Batch size is 100 and learning rate $\eta=0.001$. Note that in our in preliminary work \cite{Mei2021}, we find that the linear learning model may achieve better performance compared with other models, and therefore the linear activation function is chosen in this paper. We use $V$ and ${V'}$ to denote the input and output of the decoder, respectively. The number of iterations is illustrated through a color bar. From IP-I, it can be seen that the final value of mutual information ${\rm{I}}\left( {T;V'} \right)$ in each layer tends to be equal to the final value of ${\rm{I}}\left( {T;V} \right)$, which means that the information from $V$ has been learnt and transferred to $V'$ by each layer. In IP-II, ${\rm{I}}\left( {T';V'} \right) < {\rm{I}}\left( {T;V} \right)$ in each layer, which implies that all the layers are not overfitting. The tendency of ${\rm{I}}\left( {T;V} \right)$ to approach the value of ${\rm{I}}\left( {T';V} \right)$ can be observed from IP-III. Finally, from all the IPs, it is easy to notice that the mutual information does not change significantly when the number of iterations is larger than 200. Meanwhile, according to Fig. \ref{InformationPlane_S=4}(d), the MSE value reaches a very low value and also does not change sharply. It means that 200 iterations are enough for the task of 64-subcarrier channel estimation using a DNN with the above-mentioned topology.

In \cite{liu2019online},  single hidden layer feedforward neural network (SLFN)-based channel estimation and equalization scheme shows outstanding performance compared to the DNN-based scheme. However, it is still unknown whether a SLFN can completely learn the channel structural information from the training set, and then transfer the information from the input layer to the output layer. Fig. \ref{InformationPlane_S=1}(a), (b) and (c) illustrate the behavior of the IP-I, IP-II, and IP-III  in an SLFN-based OFDM channel estimator with topology ``$128 - 128 - 128$'' where the other hyperparameters are same to that of the DNN with $S=4$. In this case, the IP-I and IP-II are entirely identical since $S=1$. From IP-I, when the number of iterations is larger than 50, it can be seen that ${\rm{I}}\left( {T;V'} \right)$ in hidden layer tends to be equal to ${\rm{I}}\left( {T;V} \right)$. The final value of ${\rm{I}}\left( {T;V} \right)$ approximately equals to 3.5 which is nearly the same final value for $S=4$. Correspondingly, the MSE value does not change significantly. Furthermore, comparing Fig. \ref{InformationPlane_S=1}(d) to Fig. \ref{InformationPlane_S=4}(d), the MSE value decreases more rapidly and smoothly. These mean that the SLFN with 128 hidden neurons is able to learn the same information from the training set with $N_c=64$ and its learning speed and quality are better than that of the DNN with $S=4$. 
\section{Conclusion}
\label{Conclusion}
In this paper, we propose a framework to understand the manner of the DNNs in physical communication. We find that a DNN-based transmitter essentially tries to produce a good representation of the information source. In terms of convergence, the AE has specific requirements for wireless channels, i.e., the channel should be AWGN with high SNR. Then, we quantitatively analyze the MSE performance of neural network-based estimators and the information flow in neural network-based communication systems. The analysis reveals that, in the practical scenario, i.e, given limited training samples, a neural network with deeper layers may has inferior MSE performance compared to a shallow one. For the task of inference (e.g., channel estimation), we verify that the decoder can learn the information from a training set, and the shallow neural network with a single hidden layer has advantages in learning speed and quality by comparing with the DNN.

We believe that this framework has the potential for the design of DNN-based physical communication. Specifically, theoretical analysis shows that the application of the neural network-based communication system with end-to-end structure has high requirements on the channel, and the practical application range is limited. Therefore, it is more suitable to deploy neural networks at receiver. Under the condition of limited training samples, the neural network with single hidden layer can achieve the optimal MSE estimation performance, and therefore an SLFN can be deployed in a receiver for the task of channel estimation. Furthermore, the size of the training set and the dimension of a DNN can be determined by the proposed framework.

In the future, limitations for the DNN under fading channels should be solved. It would be interesting to use deep reinforcement learning technique to design waveform parameters for a transmitter instead of entirely replacing it with a DNN.

\appendices
\section{Matrix-based Functional of Rényi’s $\alpha$-Entropy}
For a random variable $X$ in a finite set $\mathcal X$, its Rényi’s entropy of order $\alpha$ is defined as
\begin{equation}
	{H_\alpha }\left( X \right) = \frac{1}{{1 - \alpha }}\log \int_{\cal X} {{f^\alpha }\left( x \right)dx} 
\end{equation}
where $f\left( x \right)$ is the PDF of the random variable $X$. Let $\left\{ {{x^{\left( i \right)}}} \right\}_{i = 1}^n$ be an \textit{i.i.d.} sample of $n$ realizations from the random variable $X$ with PDF $f\left( x \right)$. The Gram matrix ${\bf{K}}$ can be defined as ${\bf{K}}\left[ {i,j} \right] = \kappa \left( {{x_i},{x_j}} \right)$ where $\kappa :{\mathcal X} \times {\mathcal X} \mapsto {\mathbb R}$ is a real valued positive definite and infinitely divisible kernel. Then, a matrix-based analogue to Rényi’s $\alpha$-entropy for a normalized positive definite matrix $\bf{A}$ of size $n \times n$ with trace 1 can be given by the functional 
\begin{equation}
	{S_\alpha }\left( {\bf{A}} \right) = \frac{1}{{1 - \alpha }}{\log _2}\left[ {\sum\limits_{i = 1}^n {{\lambda _i}{{\left( {\bf{A}} \right)}^\alpha }} } \right]
\end{equation}
where ${{\lambda _i}\left( {\bf{A}} \right)}$ denotes the $i$-th eigenvalue of $\bf{A}$, a normalized version of $\bf{K}$:
\begin{equation}
	{\bf{A}}\left[ {i,j} \right] = \frac{1}{n}\frac{{{\bf{K}}\left[ {i,j} \right]}}{{\sqrt {{\bf{K}}\left[ {i,i} \right]{\bf{K}}\left[ {j,j} \right]} }}.
\end{equation}
Now, the joint-entropy can be defined as
\begin{equation}
	{S_\alpha }\left( {{\bf{A}},{\bf{B}}} \right) = {S_\alpha }\left[ {\frac{{{\bf{A}} \odot {\bf{B}}}}{{{\rm{tr}}\left( {{\bf{A}} \odot {\bf{B}}} \right)}}} \right].
\end{equation}
Finally, the matrix notion of Rényi’s mutual information can be defined as 
\begin{equation}
	{I_\alpha }\left( {{\bf{A}};{\bf{B}}} \right) = {S_\alpha }\left( {\bf{A}} \right) + {S_\alpha }\left( {\bf{B}} \right) - {S_\alpha }\left( {{\bf{A}},{\bf{B}}} \right).
\end{equation}




\ifCLASSOPTIONcaptionsoff
  \newpage
\fi



%

%
%

\bibliographystyle{IEEEtran}  
\bibliography{IEEEabrv,Bibliography/MyCollection}

%








\begin{IEEEbiography}[{\includegraphics[width=1in,height=1.25in,clip,keepaspectratio]{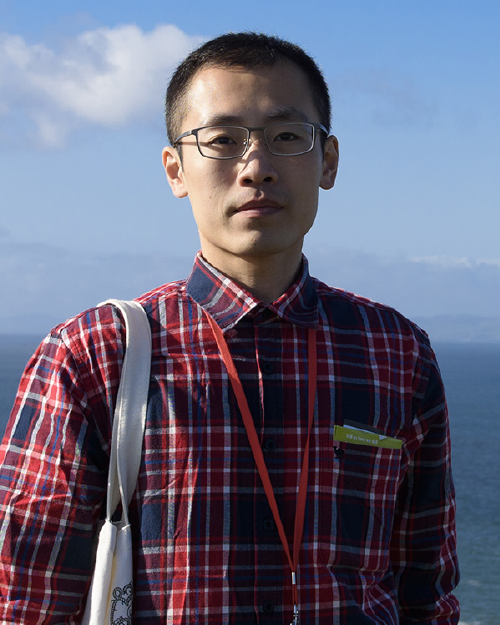}}]{Jun Liu}
	received the B.S. degree in optical information science and technology from the South China University of Technology (SCUT), Guangzhou, China, in 2015, and the M.E. degree in communications and information engineering from the National University of Defense Technology (NUDT), Changsha, China, in 2017, where he is currently pursuing the Ph.D. degree with the Department of Cognitive Communications. 
	
	He was a visiting Ph.D. student with the University of Leeds from 2019 to 2020. His current research interests include machine learning with a focus on shallow neural networks applications, signal processing for broadband wireless communication systems, multiple antenna techniques, and wireless channel modeling.
\end{IEEEbiography}

\begin{IEEEbiography}[{\includegraphics[width=1in,height=1.25in,clip,keepaspectratio]{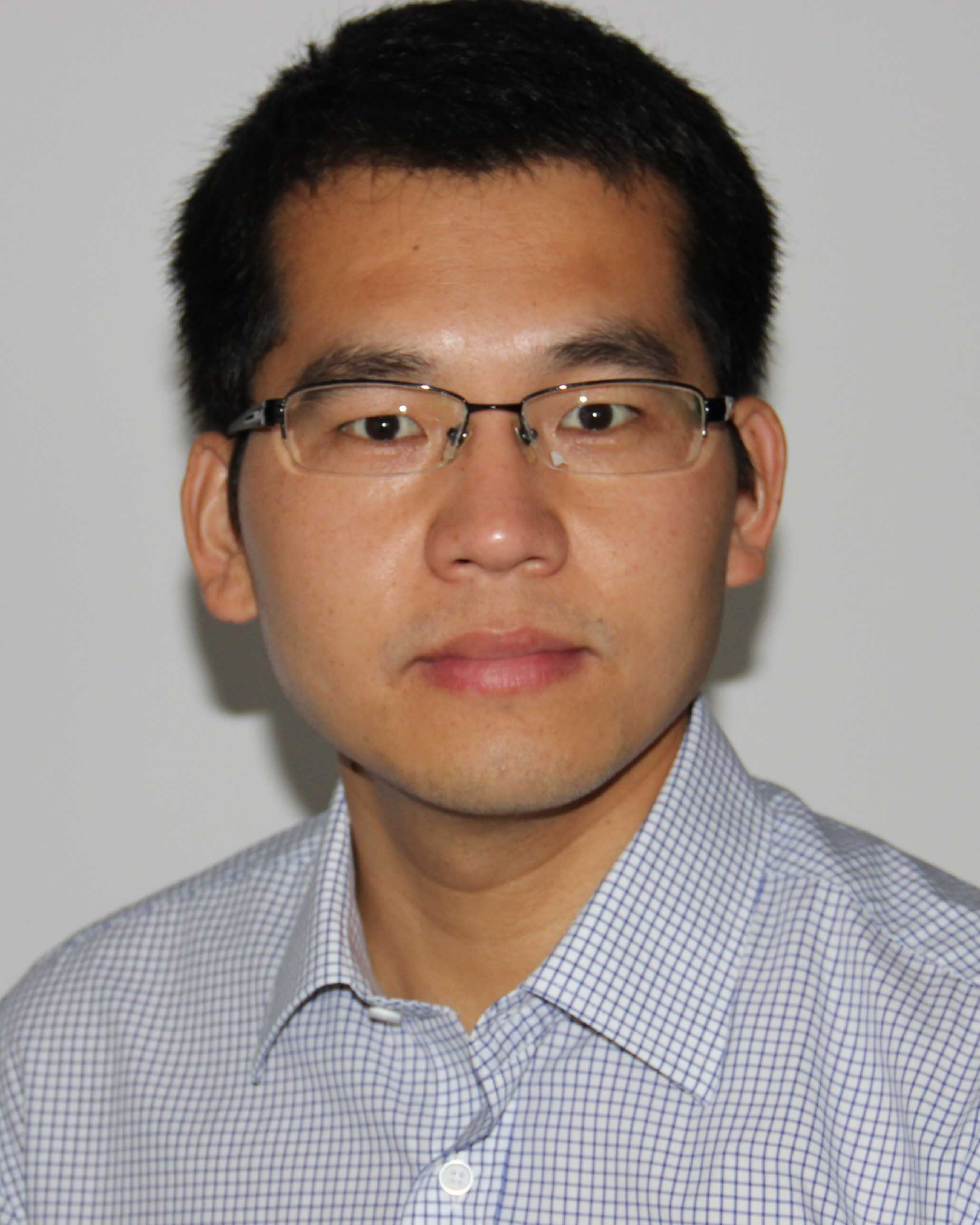}}]{Haitao Zhao}
	(Senior Member, IEEE) received his B.E., M.Sc. and Ph.D. degrees all from the National University of Defense Technology (NUDT), P. R. China, in 2002, 2004 and 2009 respectively. And he is currently a professor in the Department of Cognitive Communications, College of Electronic Science and Technology at NUDT. Prior to this, he visited the Institute of ECIT, Queen’s University of Belfast, UK and Hong Kong Baptist University. His main research interests include wireless communications, cognitive radio networks and self-organized networks. He has served as a TPC member of IEEE ICC'14-22, Globecom'16-22, and guest editor for IEEE Communications Magazine. He is a senior member of IEEE.
\end{IEEEbiography}

\begin{IEEEbiography}[{\includegraphics[width=1in,height=1.25in,clip,keepaspectratio]{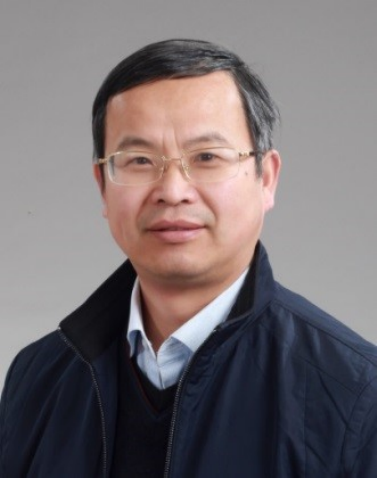}}]{Dongtang Ma}
	(SM’13) received the B.S. degree in applied physics and the M.S. and Ph.D. degrees in information and communication engineering from the National University of Defense Technology (NUDT), Changsha, China, in 1990, 1997, and 2004, respectively. From 2004 to 2009, he was an Associate Professor with the College of Electronic Science and Engineering, NUDT. Since 2009, he is a professor with the department of cognitive communication, School of Electronic Science and Engineering, NUDT. From Aug. 2012 to Feb. 2013, he was a visiting professor at University of Surrey, UK. His research interests include wireless communication and networks, physical layer security, intelligent communication and network. He has published more than 150 journal and conference papers. He is one of the Executive Directors of Hunan Electronic Institute. He severed as the TPC member of PIMRC from 2012 to 2020.
\end{IEEEbiography}

\begin{IEEEbiography}[{\includegraphics[width=1in,height=1.25in,clip,keepaspectratio]{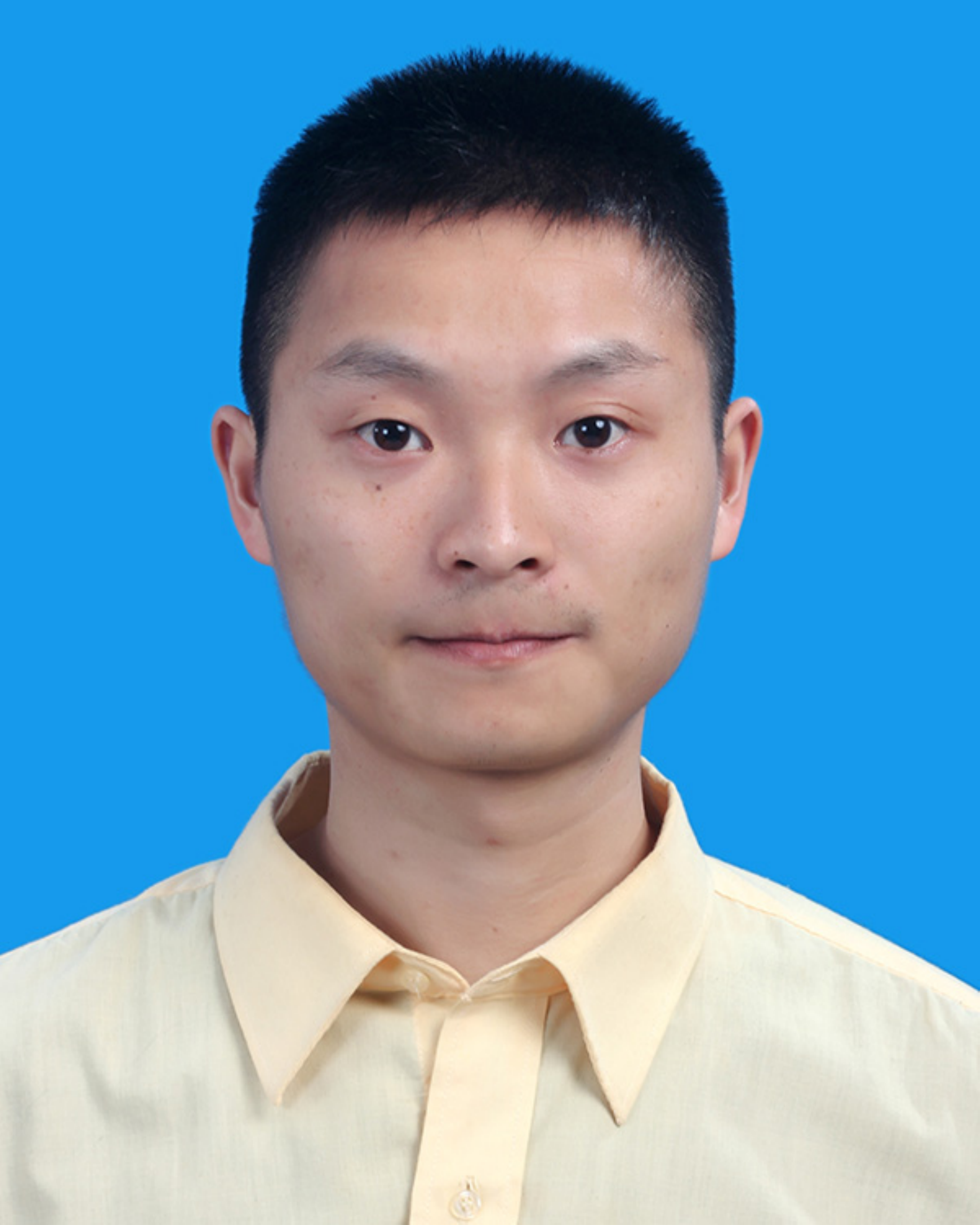}}]{Kai Mei}
	received the master’s degree from the National University of Defense Technology, in 2017, where he is currently pursuing the Ph.D. degree. His research interests include synchronization and channel estimation in OFDM systems and MIMO-OFDM systems, and machine learning applications in wireless communications.
\end{IEEEbiography}

\begin{IEEEbiography}[{\includegraphics[width=1in,height=1.25in,clip,keepaspectratio]{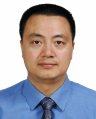}}]{Jibo Wei}
	(Member, IEEE) received the B.S. and M.S. degrees from the National University of Defense Technology (NUDT), Changsha, China, in 1989 and 1992, respectively, and the Ph.D. degree from Southeast University, Nanjing, China, in 1998, all in electronic engineering. He is currently the Director and a Professor of the Department of Communication Engineering, NUDT. His research interests include wireless network protocol and signal processing in communications, more specially, the areas of MIMO, multicarrier transmission, cooperative communication, and cognitive network. He is a member of the IEEE Communication Society and also a member of the IEEE VTS. He also works as one of the editors of the Journal on Communications and is a Senior Member of the China Institute of Communications and Electronics.
\end{IEEEbiography}

\end{document}